\pdfoutput=1
\documentclass[10pt,pra,aps,superscriptaddress,nofootinbib,twocolumn,longbibiliography]{revtex4-1}

\usepackage[compatibility=false]{caption} 
\usepackage{amsmath,bbm}
\DeclareMathOperator{\Tr}{Tr}
\usepackage{amsthm}
\usepackage{amsmath}
\usepackage{latexsym}
\usepackage{amssymb}
\usepackage{float}
\usepackage{graphicx}           
\usepackage{color}
\usepackage{xcolor}
\usepackage{mathpazo}
\usepackage{comment}
\usepackage{enumitem}
\usepackage{multirow}
\usepackage{setspace}
\usepackage{caption}
\usepackage{subcaption}
\usepackage{mathtools} 
\usepackage{extarrows} 

\usepackage[commandnameprefix=always]{changes}
\usepackage{svg}
\usepackage{subcaption}
\usepackage[colorlinks=true, linkcolor=blue, citecolor=blue, urlcolor=blue]{hyperref}

\usepackage{changes}

\newcommand{\be}{\begin{equation}}
\newcommand{\ee}{\end{equation}}
\newcommand{\bea}{\begin{eqnarray}}
\newcommand{\eea}{\end{eqnarray}}

\def\squareforqed{\hbox{\rlap{$\sqcap$}$\sqcup$}}
\def\qed{\ifmmode\squareforqed\else{\unskip\nobreak\hfil
\penalty50\hskip1em\null\nobreak\hfil\squareforqed
\parfillskip=0pt\finalhyphendemerits=0\endgraf}\fi}
\def\endenv{\ifmmode\;\else{\unskip\nobreak\hfil
\penalty50\hskip1em\null\nobreak\hfil\;
\parfillskip=0pt\finalhyphendemerits=0\endgraf}\fi}

\newcommand{\tr}{\text{Tr}}

\newcommand{\ket}[1]{|#1\rangle}
\newcommand{\bra}[1]{\langle#1|}

\newcommand{\re}{\color{blue}}  
\newcommand{\blk}{\color{black}}
\makeatletter
\newtheorem*{rep@theorem}{\rep@title}
\newcommand{\newreptheorem}[2]{%
\newenvironment{rep#1}[1]{%
 \def\rep@title{#2 \ref{##1}}%
 \begin{rep@theorem}}%
 {\end{rep@theorem}}}
\makeatother

\newtheorem{thm}{Theorem}
\newreptheorem{thm}{Theorem}

\newtheorem{definition}{Definition}

\newtheorem*{obs}{Observation}

\newtheorem{coro}{Corollary}

\begin{document}


\title{Limits of classical correlations and quantum advantages under (anti-)distinguishability constraints in multipartite communication
}


\author{Ankush Pandit}
\affiliation{School of Physics, Indian Institute of Science Education and Research Thiruvananthapuram, Kerala 695551, India}
\affiliation{Department of Physics, School of Basic Sciences, Indian Institute of Technology Bhubaneswar, Odisha 752050, India}
\author{Soumyabrata Hazra}
\affiliation{CQST and CCNSB, International Institute of Information Technology Hyderabad, Telangana 500032, India}
\author{Satyaki Manna}
\affiliation{School of Physics, Indian Institute of Science Education and Research Thiruvananthapuram, Kerala 695551, India}
\affiliation{Department of Physics, School of Basic Sciences, Indian Institute of Technology Bhubaneswar, Odisha 752050, India}
\author{Anubhav Chaturvedi}
\affiliation{Faculty of Applied Physics and Mathematics,
 Gda{\'n}sk University of Technology, Gabriela Narutowicza 11/12, 80-233 Gda{\'n}sk, Poland}
\affiliation{International Centre for Theory of Quantum Technologies (ICTQT), University of Gda{\'n}sk, 80-308 Gda\'nsk, Poland}
\author{Debashis Saha}
\affiliation{School of Physics, Indian Institute of Science Education and Research Thiruvananthapuram, Kerala 695551, India}
\affiliation{Department of Physics, School of Basic Sciences, Indian Institute of Technology Bhubaneswar, Odisha 752050, India}


\begin{abstract}
We consider communication scenarios with multiple senders and a single receiver. Focusing on communication tasks where the distinguishability or antidistinguishability of the sender's input is bounded, we show that quantum strategies—without any shared entanglement—can outperform the classical ones. We introduce a systematic technique for deriving the facet inequalities that delineate the polytope of classical correlations in such scenarios. As a proof of principle, we recover the complete set of facet inequalities for some nontrivial scenarios involving two senders and a receiver with no input. Explicit quantum protocols are studied that violate these inequalities, demonstrating quantum advantage. We further investigate the task of antidistinguishing the joint input string held by the senders and derive upper bounds on the optimal classical success probability. Leveraging the Pusey–Barrett–Rudolph theorem, we prove that when each sender has a binary input, the quantum advantage grows with the number of senders. We also provide sufficient conditions for quantum advantage for arbitrary input sizes and illustrate them through several explicit examples.

\end{abstract}


\maketitle

\re 

\blk 

\section{Introduction}
Communication tasks play a central role both in quantum information science and in foundational research. Under varied communication constraints, quantum resources and protocols routinely surpass their classical counterparts. Notable examples include dimension-constrained communication complexity problems \cite{ccbook,TRbook,rao_yehudayoff_2020,Yao,saha2023}, parity-oblivious multiplexing \cite{PhysRevLett.102.010401, Ambainis2019}, oblivious communication task \cite{Saha_2019,saha2019,Chaturvedi2021characterising, Hazra2026efficient}, and the broad family of random-access codes \cite{10.1145/1008908.1008920,10.1145/301250.301347,Saha_EPL,PhysRevLett.121.050501,Hameedi_PRA,RACA,RACA1}. Random access codes, in particular, underpin the semi-device-independent one-way quantum key distribution \cite{PhysRevA.84.010302}, and randomness certification \cite{PhysRevA.84.034301,PhysRevX.6.011020,qrandom}. Beyond these practical relevance, these tasks mirror the prepare-and-measure scenarios that are widely studied in quantum foundations. For instance, the quantum advantage observed in parity-oblivious multiplexing furnishes an operational proof of preparation contextuality in quantum theory \cite{PhysRevLett.102.010401}. 


Although most communication tasks are traditionally analyzed by placing an upper bound on the alphabet size of a classical message or the Hilbert space dimension of a quantum message, another approach is to restrict the message's distinguishability and antidistinguishability \cite{bod,PRR_2024,Tavakoli2020informationally,epi_incom,ray2024epistemicmodelexplainantidistinguishability}. Distinguishability is the maximal probability that a receiver can correctly infer the sender's input from the message, thus quantifying the information revealed. Antidistinguishability, by contrast, is the maximal probability of deliberately guessing the input incorrectly. Because these quantities vary continuously, they furnish a finer-grained description of information content than the coarse, discrete notion of dimension \cite{bod,Tavakoli2020informationally}. In cryptographic settings, where privacy is paramount, limiting distinguishability offers a natural operational gauge of how effectively the sender's input remains hidden. 
\sloppy
In multipartite communication, dimension-based studies have uncovered quantum-over-classical advantages in multipartite communication tasks, without the use of entanglement \cite{QF,Vertesi-pra,Marcin-pra,Hameedi_PRA,Saha_EPL,Chakraborty_2025}. However, it is not yet clear how quantum advantage manifests when the communication is constrained by distinguishability and antidistinguishability rather than by message dimension.

Here we introduce multipartite communication tasks in which each sender’s message is restricted solely by an upper bound on the probability that the receiver can (or cannot) infer that sender’s input—i.e., by its distinguishability or antidistinguishability. Crucially, such bounds place no a priori limit on the Hilbert-space dimension of the underlying physical system. Beyond their conceptual appeal, these correlations are experimentally accessible with current high-fidelity prepare-and-measure platforms \cite{PhysRevApplied.18.064028,PRXQuantum.3.010352}. Another important feature of our communication tasks is that the receiver has no input. While two-party communication tasks are known not to exhibit any quantum advantage in such scenarios \cite{fw}, we observe clear quantum advantages in multi-sender–single-receiver settings. It is worth emphasizing that our quantum protocol does not rely on any pre-shared entangled resources. 

In this article, we study a multipartite communication scenario with $N$ independent senders and a single receiver, where each sender’s message is constrained by a bound on either its distinguishability or its antidistinguishability. Upon receiving these messages, classical or quantum, the receiver performs a predetermined measurement to generate an output distribution. First, we develop a framework to characterize the set of achievable classical probability distributions without fixing the numerical value of distinguishability (or antidistinguishability) in advance; in other words, our description encompasses all valid classical communication protocols for any allowed distinguishability (or antidistinguishability) values. We use this method to explicitly enumerate the facet inequalities of classical polytopes on a tripartite scenario with two independent senders and a single receiver, where either distinguishability or antidistinguishability of each sender’s input bounded and the receiver’s measurement choice is fixed.

Next, we turn to the quantum case. We adopt semi-definite programming (SDP) methods introduced in \cite{Chaturvedi2021characterising,Tavakoli2022informationally,PRR_2024} to the multipartite setting and for the antidistinguishability constraint. As result, we recover several instances of quantum advantage in the form of quantum violation of the previously recovered facet inequalities. Building on these observations, we define a multipartite task whose figure of merit is the antidistinguishability of the joint inputs held by the distributed senders. Drawing inspiration from the Pusey–Barrett–Rudolph (PBR) Theorem \cite{Pusey_2012}, we construct a quantum protocol that achieves an exponential advantage over any classical strategy as the number of senders grows, when each sender has two possible inputs. 

In addition, we derive a simple sufficient condition, namely a bound on the pairwise overlaps of each sender's input states, that ensures a quantum advantage in this task for arbitrary input sizes. Finally, we show that any quantum advantage in such multipartite communication scenarios implies an epistemic incompleteness of quantum theory \cite{epi_incom}, under the assumption of preparation independence.

 \section{Multipartite communication with bounded Distinguishability or Antidistinguishability} \label{multiparty communication}


\begin{figure}
    \centering
    \includegraphics[width=1\linewidth]{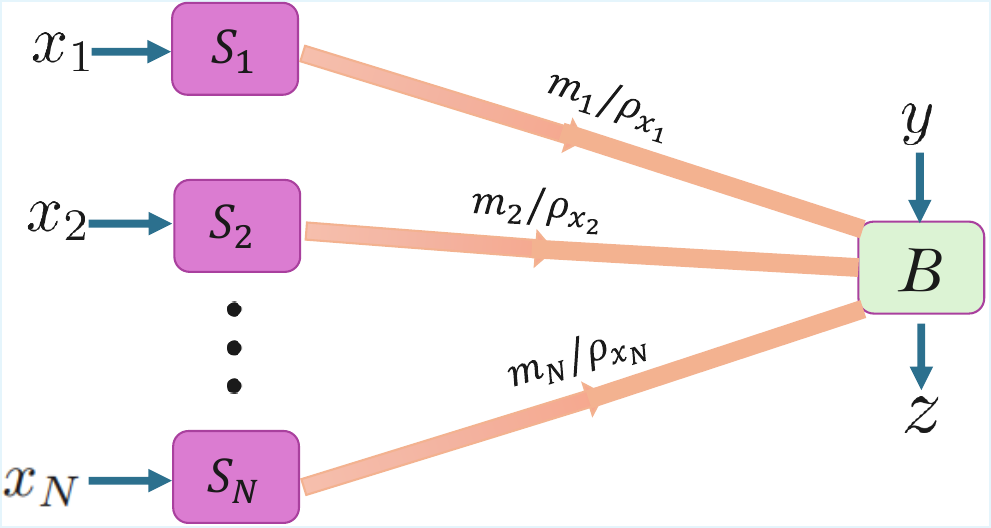}
    \caption{Schematic diagram of a multipartite communication scenario. Each sender $S_i$ receives an input $x_i$ and communicates a message, either classical or quantum ($m_i$ or $\rho_{x_i}$) to the receiver $B$. The receiver, in turn, performs some measurement based on the input $y$ and produces an outcome $z$.}
    \label{fig:enter-label}
\end{figure}
 
We consider a multipartite communication scenario that comprises $N$ senders, denoted by $S_i,~ i \in \left[N\right]$ and one receiver. Here $[N]$ represents the set $\{1,2,...,N\}$. Each sender $S_i$ receives an input $x_{i}\in \left[n_{x_{i}}\right]$ and independently communicates a message to the receiver depending on the input. The senders are not allowed to communicate among themselves. Receiver is also given an input $y\in\left[n_{y}\right]$. Depending on all the received messages and the input $y$, the receiver provides an output $z\in [n_z]$. The constraint of the task lies in the communicated message such that the probability of distinguishing (or antidistinguishing) the inputs of each sender is bounded by some value. This communication scenario gives rise to correlations between senders and the receiver which can be fully characterized by the operational probabilities of the form $p(z|\vec{x},y)$ where $\vec{x}=(x_1,...,x_N)$ describes a collection of inputs from all the senders.
It is worth noting that we consider a general scenario for the subsequent discussion in this section. However, all the explicit scenarios we analyze to demonstrate quantum advantages involve a fixed input on the receiver's end.

The distinguishability (or antidistinguishability) of each sender $S_i$'s input is upper bounded by $D_i$ (or $A_i$). 
\begin{definition}[Distinguishability of inputs] The distinguishability of a collection of inputs $\{x_i\}_{x_i=1}^{n_{x_i}}$ is said to be bounded by $D_i$ if they satisfy  \begin{align}\max_R \sum_{o_i} \max_{x_i}~ q_{x_i} p\left(o_i|x_i,R\right) \leq D_i, \label{op33} 
\end{align}
where the input $x_i$ is assumed to be sampled from the distribution $\{q_{x_i}\}_{x_i}$ and $o_i$ denotes the outcome of measurement $R$. 
\end{definition} This essentially says that, corresponding to an output $o_i$ the best guess of input $x_i$ is the one which maximizes $q_{x_i} p\left(o_i|x_i,R\right)$ over the index $x_i$. Then maximize it over all possible measurements $R$.
\begin{definition}[Antidistinguishability of inputs]\label{def:anti} The antidistinguishability of a collection of inputs $\{x_i\}_{x_i=1}^{n_{x_i}}$ is said to be bounded by $A_i$ if they satisfy \begin{align}\label{anti22}
    1-\min_R \sum_{o_i} \min_{x_i} q_{x_i}p\left(o_i|x_i,R\right) \leq A_i.
\end{align}
\end{definition} Here the strategy is, corresponding to an output $o_i$, the best guess of input $x_i$ is the one that minimizes $q_{x_i} p\left(o_i|x_i,R\right)$ over the index $x_i$. Then minimize it over all possible measurements $R$. 

Note that $\{q_{x_i}\}_{x_i=1}^{n_{x_i}}$ represents an $a priori$ probability distribution over $n_{x_i}$ inputs of $S_i$. Being a valid probability distribution $\{q_{x_i}\}_{x_i}$ satisfies $q_{x_i}\geq 0 ~\forall x_i; \sum_{x_i=1}^{n_{x_i}}q_{x_i}=1 ~\forall i$. The independence of senders implies that the probability of obtaining a joint outcome $\vec{o}=(o_1,...,o_N)$ while $\vec{R}=(R_1,...,R_N)$ is a collection of individual measurements corresponding to input $\vec{x}$, can be decomposed as a product of individual outcome probabilities $\{p(o_i|x_i,R_i)\}$: \begin{align}\label{independence of senders}
    p(\vec{o}|\vec{x},\vec{R})=\prod_{i=1}^N p(o_i|x_i,R_i).
\end{align} 
It is important to note that, performing local measurements $R_i$ is only required to describe the constraints of the communication task. However, the receiver is allowed to perform global measurements. The optimal measurements that saturate the inequalities \eqref{op33}, and \eqref{anti22} may not be part of the receiver's choice of measurement. The distinguishability \eqref{op33} and antidistinguishability constraints \eqref{anti22} can be expressed in terms of the operational probabilities as  \begin{align}
    \max_y\sum_z \max_{x_i}~ q_{x_i} p\left(z|x_i,y\right) &\leq D_i ~\forall i ,\\ \text{and~~}~
    1- \min_y\sum_z \min_{x_i} q_{x_i}p\left(z|x_i,y\right) &\leq A_i ~\forall i, 
\end{align} respectively. The independence of senders asserts the following to be true
\begin{align}
    \max_y\sum_z \max_{x_i}~ q_{x_i} p\left(z|\vec{x},y\right) &\leq D_i ~\forall i ,\label{op3}\\ \text{and~~}~
    1- \min_y\sum_z \min_{x_i} q_{x_i}p\left(z|\vec{x},y\right) &\leq A_i ~\forall i. \label{anti2}
\end{align} Fundamentally, the operational probabilities satisfy positivity and normalization conditions: 
\begin{align}
    p\left(z|\vec{x},y\right) &\geq 0, ~\forall z,\vec{x},y \label{op1},\\
    \sum_z p\left(z|\vec{x},y\right) &= 1, ~\forall \vec{x},y . \label{op2}
\end{align}

The set of correlations $p\left(z|\vec{x},y\right) $ satisfying Eqs. \eqref{op3},\eqref{op1}, and \eqref{op2} describes a convex polytope $\mathbb{D}$ that does not assume any underlying theory of the communication model. Analogously, the set of correlations satisfying Eqs. \eqref{anti2}-\eqref{op2} forms a convex polytope $\mathbb{A}$. It is possible to construct a more general set of correlations where $D_i\left(A_i\right)~\forall i$ are treated as bounded variables instead of a constant. Distinguishability variables can be bounded by 
\begin{align}\label{op4}
    \max_{x_i} q_{x_i} \leq D_i \leq 1, ~ \forall i.
\end{align} 
An ensemble of $n_{x_i}$ possible choices for $x_i$ can always be distinguished with probability atleast $\max_{x_i} q_{x_i}$ with a strategy always to output a fixed index that is assigned highest probability according to the distribution $q_{x_i}$. Antidistinguishability variables can be bounded by \begin{align}\label{anti3}
    1-\min_{x_i} q_{x_i} \leq A_i \leq 1, ~
    \forall i.
\end{align} 
An ensemble of $n_{x_i}$ possible choices for $x_i$ can always be antidistinguished with probability atleast $1-\min_{x_i} q_{x_i}$ with a strategy always to output a fixed index that is assigned lowest probability according to the distribution $q_{x_i}$. \begin{definition}[Operational distinguishability (or antidistinguishability) polytope]
    The operational probabilities satisfying the constraints \eqref{op3}, and \eqref{op1}-\eqref{op4} (or \eqref{anti2}-\eqref{op2}, and \eqref{anti3}) describe a convex polytope $\mathbb{D}^+ $ (or $\mathbb{A}^+)$ such that any valid choices of the variables $D_i (\text{~or~} A_i)$s in $\mathbb{D}^+$ (or $\mathbb{A}^+)$ reduce it to a corresponding polytope $\mathbb{D}$ (or $\mathbb{A})$.
\end{definition}  
Essentially, $\mathbb{D}^+$ and $\mathbb{A}^+$ encompass all the polytopes $\mathbb{D}$ and $\mathbb{A}$ respectively, for all valid fixed choices of the variables $D_i$ and $A_i$. The convexity of $\mathbb{D}^+$ can be easily verified as $\alpha\left(\{p'(z|\vec{x},y)\}, \vec{D}'\right)+(1-\alpha)\left(\{p''(z|\vec{x},y)\}, \vec{D}''\right) \in \mathbb{D}^+$ if $\left(\{p'(z|\vec{x},y)\}, \vec{D}'\right)$ and $\left(\{p'(z|\vec{x},y)\}, \vec{D}'\right)$ belongs to $\mathbb{D}^+$ where $\alpha \in [0,1]$. 

Having introduced the general framework, we analyze the scenario from the perspective of classical communication. In the following section, we provide a complete picture of the set of classical correlations attainable in multipartite communication scenarios constrained by the distinguishability (or antidistinguishability) of senders' input.

\subsection{Characterizing the set of multipartite classical correlations} \label{secA}
Each sender $S_i$ encodes a message $m_i$ based on the input 
index $x_i\in[n_{x_i}]$ according to an encoding probability $\{p_e\left(m_i|x_i\right)\}$. The encoding probabilities fundamentally satisfy the positivity and normalization condition,
\begin{align}
    p_e\left(m_i|x_i\right) &\geq 0  ~\forall m_i,x_i,i ,\label{en1}\\ \sum_{m_i}p_e\left(m_i|x_i\right) &= 1 ~\forall x_i,i. \label{en2}
\end{align}
The receiver obtains an outcome $z\in[n_z]$ conditioned on receiving an $N$-tuple message $\vec{m}=(m_1,...,m_N)$ and measurement choice $y\in [n_y]$, inducing a conditional decoding probability distribution $\{p_d\left(z|\vec{m},y\right)\}$. Respecting the positivity and normalization condition, the decoding probabilities satisfy
\bea
    p_d\left(z|\vec{m},y\right) \geq 0, &&\forall z,\vec{m},y, \label{de1}\\ 
    \sum_z p_d\left(z|\vec{m},y\right) = 1, &&\forall \vec{m},y .\label{de2}
\eea
The classical correlations can be characterized using the following decomposition of 
operational probabilities, 
\begin{align}\label{pcl}
    p\left( z|\vec{x},y\right)=\sum_{\vec{m}} p_d\left(z|\vec{m},y\right)\prod_{i=1}^N p_e\left(m_i|x_i\right).
\end{align} 
More generally, corresponding to each sender, we can write \begin{align}\label{pcl_mg}
    p(z_i|x_i,M) = \sum_{m_i}p_d(z_i|m_i,M)p_e(m_i|x_i),
\end{align} where the decoding probabilities are considered for all possible $n_{x_i}$-outcome measurement $M$. Note that, here $M$ is not necessarily a part of the receiver's choice of measurements.
Using this expression in the distinguishability constraint \eqref{op33}, we obtain 
\begin{align}\label{dpzmxy}
    \max_M\sum_{z_i}\max_{x_i}q_{x_i}\sum_{m_i}p_d(z_i|m_i,M)p_e(m_i|x_i) \leq D_i .
\end{align}
Note that for each choice of $m_i$ the decoding probabilities satisfying \eqref{de1}, and \eqref{de2} form a convex polytope whose extremal points are characterized by deterministic decoding strategies, i.e., $p_d\left(z_i|m_i,M\right) =0~\forall z_i$ except for a specific output $z_i^*$ such that $p_d\left(z_i^*|m_i,M\right) =1$. For each $m_i$ opting for a deterministic decoding strategy, the optimization over $x_i$ can be realized by maximizing 
    $q_{x_i}p_e\left(m_i|x_i\right)$ over the index $x_i$. This leads to a simpler reformulation of \eqref{dpzmxy} as 
    \begin{align}\label{dis}
    \sum_{m_i}\max_{x_i}~ q_{x_i} p_e\left(m_i|x_i\right) &\leq D_i, ~\forall i .
\end{align}
Similarly substituting \eqref{pcl_mg} in the antidistinguishability constraint \eqref{anti22}, we obtain \begin{align}\label{dazmxy}
    1-\min_M\sum_{z_i}\min_{x_i}q_{x_i}\sum_{m_i} p_d\left(z_i|m_i,M\right)p_e\left(m_i|x_i\right) \leq A_i .
\end{align} 
    Choosing the same deterministic decoding strategy, optimization over $x_i$ can be realized by minimizing
$q_{x_i}p_e\left(m_i|x_i\right)$ over the index $x_i$ for each $m_i$. This strategy leads to a simpler expression of \eqref{dazmxy} as
\begin{align} 
    1-\sum_{m_i} \min_{x_i}~ q_{x_i} p_e\left(m_i|x_i\right) \leq A_i ~\forall i . \label{antii}
\end{align}
Moving to a simpler scenario considering $\{q_{x_i}\}_{x_i}$ to be a uniform distribution over the $n_{x_i}$ possible choices, \eqref{dis} and \eqref{antii} can be rephrased respectively as \begin{align}\label{diseq}
  \sum_{m_i} \max_{x_i}~ p_e\left(m_i|x_i\right) \leq n_{x_i}D_i, ~\forall i , \\ n_{x_i}-\sum_{m_i} \min_{x_i}~  p_e\left(m_i|x_i\right) \leq n_{x_i} A_i, ~\forall i .  
\end{align} For all the multipartite scenarios studied in Sec.\ref{sec3} we consider uniform distribution of inputs, i.e., $q_{x_i}=1/n_{x_i}$. However, for the rest of the discussion in this section, $\{q_{x_i}\}$ can be any arbitrary probability distribution.

\begin{definition}[Classical distinguishability (or antidistinguishability) polytope]\label{def5}
The set of correlations arising from the assumption of classical theory as the underlying physical process to explain the operational probabilities of $\mathbb{D}(\text{or~}\mathbb{A})$ manifests a convex polytope $\mathbb{D}_C(\text{or~}\mathbb{A}_C)$. A $\left(n_zn_y\prod_{i=1}^Nn_{x_i}\right)$-tuple $p(z|\vec{x},y)\in\mathbb{D}_C(\text{or~}\mathbb{A}_C)$ if the underlying encoding and decoding variables satisfy the constraints \eqref{en1}--\eqref{pcl} and \eqref{dis}[\text{or~}\eqref{en1}--\eqref{pcl} and \eqref{antii}].   
\end{definition}
In general $\mathbb{D}_C \subseteq \mathbb{D} \text{~and~}\mathbb{A}_C \subseteq \mathbb{A}$ as $\mathbb{D}_C \text{~and~}\mathbb{A}_C$ imposes further constraints on the communication strategies of senders.
\begin{definition}[Extended classical distinguishability (or antidistinguishability) polytope]
    Similar to the extension of $\mathbb{D}(\text{or~}\mathbb{A})$ to $\mathbb{D}^+ (\text{or~}\mathbb{A}^+)$, the classical polytope $\mathbb{D}_C (\text{or~}\mathbb{A}_C)$ can be extended to $\mathbb{D}_C^+ (\text{or~} \mathbb{A}_C^+)$ by considering all $D_i (\text{or~} A_i)$s as bounded variables as described in \eqref{op4} [or \eqref{anti3}]. A $\left(N+n_z n_y\prod_{i=1}^N n_{x_i}\right)$-tuple $\left(\{p\left(z|\vec{x},y\right)\}, \vec{D} \right) \in \mathbb{D}_C^+\left[\text{or~}\left(\{p(z|\vec{x},y)\},\vec{A}\right)\in\mathbb{A}_C^+\right]$ if the underlying encoding and decoding probabilities satisfy the respective constraints of Def. \ref{def5} along with the bound \eqref{op4} [or \eqref{anti3}].
\end{definition}
 Here $\vec{D}=(D_1,...,D_N)$ and $\vec{A}=(A_1,...,A_N)$ respectively represents a collection of distinguishabilities and antidistinguishabilities of all the senders. 

The polytope $\mathbb{D}_C^+ \left(\text{~or~}\mathbb{A}_C^+\right)$ essentially emerges from product of two types of polytope, one associated with the encoding variables $\left(p_e\left(m_i|x_i\right), D_i\right)\left[\text{~or~}\left(p_e\left(m_i|x_i\right), A_i\right)\right]$ and another with the decoding variables $p_d\left(z|\vec{m},y\right)$. 
\begin{definition}[Encoding distinguishability polytope]
    Encoding variables $p_e\left(m_i|x_i\right)$ along with distinguishability variables $D_i$, satisfying a set of linear equalities and inequalities (\ref{op4}),(\ref{en1}),(\ref{en2}), and (\ref{dis}) form a polytope $\mathbb{E}_D$. 
\end{definition} Any interior point $\left(\vec{p}_e, \vec{D}\right)$ of $\mathbb{E}_D$ can be written as a convex mixture of its extremal points, labeled by $v$: 
\begin{align}
    \left(\vec{p}_e, \vec{D}\right) = \sum_v \lambda_v \left( \vec{p}_{e_v},\vec{D}_v\right),  \\ \text{where}~~ \lambda_v > 0, ~\forall v, ~\sum_v \lambda_v = 1.
\end{align} Here $\vec{p}_e=(p(m_1=1|x_1),\cdots,p(m_1=n_{m_1}|x_1),p(m_2=1|x_2),\cdots,p(m_N=n_{m_N}|x_N))$ is a $\left(\sum_in_{m_i}n_{x_i}\right)$-tuple of all the encoding probabilities where $n_{m_i}$ represents no. of possible values of the message $m_i$. 
\begin{definition}[Encoding antidistinguishability polytope]
    Encoding variables along with antidistinguishability variables $A_i$, satisfying a set of linear equalities and inequalities (\ref{anti3}),(\ref{en1}),(\ref{en2}),(\ref{antii}) form a polytope $\mathbb{E}_A$ .
\end{definition} Any interior point $\left(\vec{p}_e^{\prime}, \vec{A}\right)$ of $\mathbb{E}_A$ can be written as a convex mixture of its extremal points, labeled by $w$: \begin{align}
    \left(\vec{p}_e^{\prime}, \vec{A}\right) = \sum_w \lambda_w \left( \vec{p}_{e_w}^{\prime},\vec{A}_w\right),  \\ \text{where}~ \lambda_w > 0 ~\forall w, ~\sum_w \lambda_w = 1.
\end{align} Here $\vec{p}_e^{\prime}$ is a $\left(\sum_in_{m_i}n_{x_i}\right)$-tuple of all the encoding probabilities analogous to $\vec{p}_e$. 

Both the encoding polytopes $\mathbb{E}_D$ and $\mathbb{E}_A$ entail $\left(\sum_{i=1}^N n_{m_i}n_{x_i} + N\right)$ variables. In order to obtain the polytopes $\mathbb{E}_D$ and $\mathbb{E}_A$ we need an upper bound on $n_{m_i}$. Following the method in Appendix C of \cite{Tavakoli2022informationally}, it is sufficient to consider $n_{m_i}=2^{\left(n_{x_i}-1\right)}$. The proof given for the distinguishability constraint can be easily shown to be true for antidistinguishability as well. 

\begin{definition}[Decoding polytope]
The polytope $\mathbb{M}$ associated with variables $p_d\left(z|\vec{m},y\right)$ is defined by the set of inequalities and equalities as described in (\ref{de1}) and  (\ref{de2}).    
\end{definition}
 For each specific choice of an $N$-tuple message $\vec{m}^*$, the decoding probabilities $\{p_d(z|\vec{m}^*,y)\}$ satisfying \eqref{de1},\eqref{de2}
induces $n_z^{n_y}$ extremal points, each corresponding to a deterministic outcome strategy. Following the same decoding strategy for all the choices of message variable $\vec{m}$, ultimately generates $n_z^{n_y\prod_i n_{m_i}}$ extremal points that characterizes $\mathbb{M}$, where $\prod_in_{m_i}=n_{\vec{m}}$ is the total number of choices for $\vec{m}$. Any interior point $\vec{p}_d$ of $\mathbb{M}$ can be written as a convex mixture of its extremal points, labeled by $u$: \begin{align}
    \vec{p}_d = \sum_u \nu_u \vec{p}_{d_u} ~, \\ \text{where~} \nu_u > 0,~ \sum_u \nu_u = 1 .
\end{align} Here $p_{d}=(p(z=1|m=1,y=1),...,p(z=n_z|m=1,y=1),p(z=1|m=2,y=1),...,p(z=n_z|m=n_{\vec{m}},y=1),...,p(z=n_z|m=n_{\vec{m}},y=n_y))$ is a $(n_zn_yn_{\vec{m}})$-tuple consisting of all the decoding probabilities.
All the extremal points of $\mathbb{D}_C^+ \left(\text{~or~}\mathbb{A}_C^+ \right)$ can be generated by multipying each extremal point of $\mathbb{E}_D\left(\text{~or~}\mathbb{E}_A\right)$ with each of  $\mathbb{M}$ according to Eq. (\ref{pcl}). The polytope $\mathbb{D}_C^+ \left(\text{~or~}\mathbb{A}_C^+ \right)$ can be equivalently characterized by its constituent extremal points or by its facet inequalities. These two representations are equivalent to one another, as given the facet inequalities, one can solve the \textit{vertex enumeration problem}, conversely, if given the vertex description, one can solve the dual \textit{facet enumeration problem}. 
For our analysis it is useful to work in the facet representation. The general form of facet inequalities of the polytopes $\mathbb{D}_C^+$ and $ \mathbb{A}_C^+$ can be written as \begin{align}
    \sum_{\vec{x},y,z} c_{\vec{x},y,z}p\left(z|\vec{x},y\right) \leq  f_{D}(\{D_i\}),  \label{facetd} \\ \text{and}~ \sum_{\vec{x},y,z} c_{\vec{x},y,z}p\left(z|\vec{x},y\right) \leq  f_{A}(\{A_i\}).  \label{facetan}
\end{align} respectively. Here $c_{\vec{x},y,z} \in \mathbb{R}$ and the classical bounds $f_{D}(\{D_i\})$ and $f_{A}(\{A_i\})$ are linear functions of individual distinguishabilities and  antidistinguishabilities respectively.

 However, given a set of distinguishabilities $\{D_i\}_i$ (or antidistinguishabilities $\{A_i\}_i$) with the associated probability distribution $\{q_{x_i}\}$, one can consider a general form of the figure of merit of the communication task as a linear function of the probabilities,
\begin{equation} \label{genI}
 \sum_{\vec{x},y,z} c_{\vec{x},y,z}p\left(z|\vec{x},y\right),
\end{equation}
which may not be a facet inequality of the polytope $\mathbb{D}_C$ (or $\mathbb{A}_C$).

\begin{definition}
    The best classical value of the expression \eqref{genI}, denoted by $\mathcal{S_C}$, is defined as
    \begin{equation} \label{def-SC}
        \mathcal{S_C} = \max_{\substack{\{p_e(m_i|x_i)\} \\ \{p_d(z|\vec{m},y)\}}} \sum_{\vec{x},y,z} c_{\vec{x},y,z} \left( \sum_{\vec{m}} p_d\left(z|\vec{m},y\right)\prod_{i=1}^N p_e\left(m_i|x_i\right) \right),
    \end{equation}
    under the constraints on $\{p_e(m_i|x_i)\}$ given by ~\eqref{dis} (or ~\eqref{antii}).
\end{definition}
The value of $\mathcal{S_C}$ can be obtained by evaluating the expression \eqref{genI} over all the external points of the polytope $\mathbb{D}_C$ (or $\mathbb{A}_C$).

\begin{definition}
    The maximum value of the expression \eqref{genI} in an operational theory is denoted by $\mathcal{S_O}$. Here, the operational theory respects only the constraints of either distinguishability \eqref{op3} or antidistinguishability \eqref{anti2} and independence of senders \eqref{independence of senders}.
 \end{definition}
Note that $\mathcal{S_O}$ may be less than the algebraic maximum of the expression \eqref{genI}, which can be obtained by relaxing the constraint of independence \eqref{independence of senders}.
In the following section, we discuss quantum communication in a multipartite setting, constrained by distinguishability (or antidistinguishability) of the senders' input.

\subsection{Multipartite quantum communication} \label{sub quant adv}
In quantum communication, each input $x_i$ of the sender $S_i$ sends a quantum system described by a density matrix $\rho_{x_i}$ and the measurement at the receiver's end, corresponding to the measurement choice $y$, is described by a set of POVM elements $\{M_{z|y}\}$ such that $\sum_z M_{z|y} = \mathbb{I}, ~\forall y$. The dimension of $\rho_{x_i}$ can be different across different senders. The observed probability acquires the form $ p(z|\vec{x},y)=\Tr\left[\left(\bigotimes_{i=1}^N\rho_{x_i}\right) M_{z|y}\right]$. The distinguishability \eqref{op33} and antidistinguishability \eqref{anti22} of quantum states $\{\rho_{x_i}\}_{x_i=1}^{n_{x_i}}$ of the $i$th sender are constrained respectively as \begin{align}
   \max_R\sum_{o_i}\max_{x_i}q_{x_i}\Tr\left(\rho_{x_i}R_{o_i}\right) \leq D_i, \label{qdis}\\ \text{and}~ 1-\min_R  \sum_{o_i}\min_{x_i}q_{x_i}\Tr\left(\rho_{x_i}R_{o_i}\right) \leq A_i ,\label{qanti} 
\end{align}where the optimization is performed over all possible $n_{o_i}$-outcome measurement $R =\{R_{o_i}\}$. Here, 
$o_i$ denotes the outcome of the measurement $R$ for the sender $S_i$. Again, the optimal POVMs that saturate \eqref{qdis} or 
\eqref{qanti} need not be part of the receiver's choice of measurement. In terms of the operational probabilities we can write \eqref{qdis} and \eqref{qanti} as \begin{align}
    \max_{y}\sum_{z}\max_{x_i}q_{x_i}\Tr\left(\otimes_{j=1}^N\rho_{x_j}{M_{z|y}}\right) \leq D_i, \\ \text{and}~ 1-\min_y  \sum_{z}\min_{x_i}q_{x_i}\Tr\left(\otimes_{j=1}^N\rho_{x_j}M_{z|y}\right) \leq A_i ,
\end{align}since a global POVM does not enhance the (anti)distinguishability of local quantum states.

\begin{definition}[Quantum (anti)distinguishability set]
The set of correlations arising from the assumption of quantum theory as the underlying physical process in order to explain the operational probabilities of $\mathbb{D}(\text{or~~}\mathbb{A})$ manifests a convex set $\mathbb{D}_Q(\text{or~}\mathbb{A}_Q)$.     
\end{definition}
Generally $\mathbb{D}_C\subseteq \mathbb{D}_Q \subseteq \mathbb{D}\text{~and~} \mathbb{A}_C \subseteq \mathbb{A}_Q\subseteq \mathbb{A}$ as the quantum set always contains the classical set and the theory-independent set of correlations can still be larger than the quantum set. The strictness of these inclusions relies on the particular kind of communication scenario being studied. We will characterize the inclusions for a few scenarios in Sec. \ref{sec3}.

\begin{definition}
The best quantum value of a linear figure of merit \eqref{genI}, denoted by $\mathcal{S_Q}$, provided a set of distinguishabilities $\{D_i\}_i$ (or antidistinguishabilities $\{A_i\}_i$) with the associated probability distribution $\{q_{x_i}\}$, is given by 
\begin{equation} \label{def-SQ}
    \mathcal{S_Q} = \max_{\{\rho_{x_i}\},\{M_{z|y}\}}~\sum_{\vec{x},y,z}c_{\vec{x},y,z}\Tr\left[\left(\otimes_{i=1}^N\rho_{x_i}\right) M_{z|y}\right] ,
\end{equation}
subject to the constraint that the states $\{\rho_{x_i}\}_{x_i}$ satisfy \eqref{qdis} [or \eqref{qanti}].
\end{definition}


\subsubsection{Semidefinite optimization to obtain a lower bound on $\mathcal{S_Q}$} \label{SDP}

 In general, it is difficult to obtain $\mathcal{S_Q}$.  
In order to identify quantum advantage in a multipartite communication scenario, it is necessary to find a lower bound of $\mathcal{S_Q}$. The method of obtaining a lower bound of $\mathcal{S_Q}$ in a bipartite communication scenario equipped with distinguishability constraints on senders' inputs was introduced in Ref. \cite{Tavakoli2022informationally}. We introduce an iterative optimization algorithm to obtain lower bounds of $\mathcal{S_Q}$ for both distinguishability and antidistinguishability based on multipartite communication tasks. The task is to perform the maximization described in \eqref{def-SQ} following the distinguishability \eqref{qdis} and antidistinguishability constraints \eqref{qanti} in respective cases. Specifically, we introduce the method for a scenario comprising two senders with $n_{x_1}$ and $n_{x_2}$ possible inputs and one receiver with no external input, i.e., a fixed choice of $y$. However, it can be generalized to arbitrary senders analogously. We discuss the SDPs corresponding to distinguishability- and antidistinguishability- constrained scenarios separately. 

\textit{Distinguishability constrained communication task:} 
The goal is to maximize \begin{align}
    \sum_{x_{1},x_{2},z} c_{x_1,x_2,z} \Tr\left[(\rho_{x_{1}} \otimes \rho_{x_{2}}) M_{z}\right] , \label{dconsdp} 
\end{align} concerning quantum states and measurements that satisfy the distinguishability constraints in \eqref{qdis} for $i\in\{1,2\}$. 
    Introducing two auxilliary variables $\sigma_i,i\in \{1,2\}$ satisfying \begin{align}
        \sigma_i\geq q_{x_i}\rho_{x_i},\label{auxsigma}
    \end{align} relieves us of the quadratic constraints in \eqref{qdis}. We can bound the distinguishability as \begin{align}
        \max_{\{M_{z_i}\}}\sum_{z_i}\max_{x_i}q_{x_i}\Tr\left(\rho_{x_i}M_{z_i}\right) \leq \max_{\{M_{z_i}\}}\sum_{z_i}\Tr\left(\sigma_iM_{z_i}\right)=\Tr(\sigma_i),
    \end{align} where we got relieved of the maximization over $M_{z_i}$ using the normalization condition $\sum_{z_i}M_{z_i}=\mathbb{I}$. Now, the constraint $\Tr(\sigma_i)\leq D_i$ along with \eqref{auxsigma} are alternative linear constraints to \eqref{qdis}. We can now formally state the SDP with linear constraints as \begin{align} 
        \max_{\rho_{x_1},\rho_{x_2},\sigma_1,\sigma_2,M_z} \sum_{x_{1},x_{2},z} c_{x_1,x_2,z} \Tr\left[(\rho_{x_{1}} \otimes \rho_{x_{2}}) M_{z}\right] \nonumber\\ \text{such that ~} \rho_{x_i}\geq 0,\Tr(\rho_{x_i})=1,~\forall i\in\{1,2\}, \label{sdpd}\\ \sigma_i\geq q_{x_i}\rho_{x_i},\Tr(\sigma_i)\leq D_i;\nonumber \\ M_z\geq 0 ~\forall z, \sum_z M_z = \mathbb{I}.\nonumber
    \end{align}
\textit{Antidistinguishability-constrained communication task :}\label{anti-sdpp}
Here again, the goal is to obtain the maximum of \eqref{dconsdp} optimizing over quantum states and measurements that satisfy the antidistinguishability constraints of \eqref{qanti} for $i\in \{1,2\}$. Again, the constraint in \eqref{qanti} is quadratic and is difficult to solve via SDP. We introduce another pair of auxiliary variables $\omega_i,i\in \{1,2\}$ which satisfies \begin{align}
    q_{x_i}\rho_{x_i}\geq \omega_i.\label{auxomega}
\end{align} We can bound the antidistinguishability as \begin{align}
    1-\min_{\{M_{z_i}\}}\sum_{z_i}\min_{x_i}q_{x_i}\Tr\left(\rho_{x_i}M_{z_i}\right)&\leq 1-\min_{\{M_{z_i}\}}\sum_{z_i}\Tr\left(\omega_iM_{z_i}\right) \nonumber\\ &=1-\Tr\left(\omega_i\right).
\end{align} Now, the constraint $1-\Tr\left(\omega_i\right)\leq A_i$ along with \eqref{auxomega} represents linear antidistinguishability constraints. We formulate the SDP as \begin{align}
    \max_{\rho_{x_1},\rho_{x_2},\omega_1,\omega_2,M_z} \sum_{x_{1},x_{2},z} c_{x_1,x_2,z} \Tr\left[(\rho_{x_{1}} \otimes \rho_{x_{2}}) M_{z}\right] \nonumber\\ \text{such that ~}~\rho_{x_i}\geq 0,\Tr(\rho_{x_i})=1,  ~\forall i\in\{1,2\}, \\ q_{x_i}\rho_{x_i}\geq\omega_i ,1-\Tr(\omega_i)\leq A_i;\nonumber \\ M_z\geq 0 ~\forall z, \sum_z M_z = \mathbb{I}.\nonumber
\end{align} 

Now we discuss how the iterative optimization algorithm (or 'SeeSaw') works to produce a lower bound on $\mathcal{S_Q}$. We discuss the algorithm only for the SDP corresponding to the distinguishability constraint. The SDP corresponding to antidistinguishability follows analogously. The algorithm samples random quantum states $\{\rho_{x_i}\}$. Then optimizes over quantum measurement $M_{z}$ respecting the relevant constraints on measurement in \eqref{sdpd}, to maximize the linear function of probabilities in \eqref{sdpd}, with fixed states $\rho_{x_i}$. In the next step of the iteration, it fixes $M_{z}$ obtained in the previous iteration and optimizes over $\rho_{x_1}$ while satisfying the relevant constraints in \eqref{sdpd}. In this iteration $\rho_{x_2}$ remains fixed. The next iteration consists of optimizing $\rho_{x_2}$ satisfying relevant constraints in \eqref{sdpd}, keeping $M_{z},\rho_{x_1}$ fixed as obtained in previous iterations. The steps are then repeated until the optimizer saturates at a specific value. We repeat this whole process many number of times and take the best value among them. This three-fold optimization algorithm optimizes from inside the quantum set, thus producing a lower bound to the quantum set.  

\begin{definition}
    The quantum value of the figure of merit \eqref{genI} obtained from the aforementioned "SeeSaw" method taking $d$-dimensional quantum states for each party (i.e., $\rho_{x_i}$ acts on $\mathbbm{C}^d$) is denoted by $\mathcal{S}_{\mathcal{Q}_d}$.
\end{definition}
In general, $\mathcal{S}_{\mathcal{Q}_d} \leqslant \mathcal{S_Q}$, for all $d$ and $\mathcal{S_C}\leqslant\mathcal{S_Q}\leqslant \mathcal{S_O}$.

\subsubsection{Quantification of quantum advantage} \label{q_advantage}

Quantum advantage over classical strategies can be quantified in several ways. In this work, we adopt two standard approaches.

First, when the values of distinguishabilities (or antidistinguishabilities) $\{D_i\}_i$ (or $\{A_i\}_i$) are specified, if the best quantum value of the figure of merit \eqref{genI} is greater than the best classical value, i.e., $\mathcal{S_Q}>\mathcal{S_C}$, that indicates an advantage. Accordingly, quantum advantage can be quantified by the ratio $\mathcal{S_Q}/\mathcal{S_C}$. Generally, $\mathcal{S_Q} \geqslant \mathcal{S_C}$ since the set of quantum correlations encompasses the classical set. 
Although obtaining the exact value of $\mathcal{S_Q}$ may be challenging, it suffices to demonstrate that $\mathcal{S}_{\mathcal{Q}_d}/\mathcal{S_C} > 1$, which serves as a lower bound for $\mathcal{S_Q}/\mathcal{S_C}$.

As an alternative approach, one may consider the minimum amount of distinguishability (or antidistinguishability) required to achieve a specific value of the figure of merit \eqref{genI}, denoted by $\mathcal{S}$. We formally define this quantity for both classical and quantum communication settings.

\begin{definition} The total amount of distinguishability (or antidistinguishability) required in a multipartite classical communication setting to achieve a given value $\mathcal{S}$ of the figure of merit \eqref{genI} is defined as: \begin{equation} D_C^{\text{total}} = \min_{\{D_i\}_i} \prod_{i=1}^N D_i \quad \left( A_C^{\text{total}} = \min_{\{A_i\}_i} \prod_{i=1}^N A_i \right), \end{equation} 
subject to the constraint $\mathcal{S_C} = \mathcal{S}$. \end{definition} 
An analogous expression is defined for quantum communication:

\begin{definition} The total amount of distinguishability (or antidistinguishability) required in a multipartite quantum communication setting to achieve a value $\mathcal{S}$ of the figure of merit \eqref{genI} is defined as: 
\begin{equation} 
D_Q^{\text{total}} = \min_{\{D_i\}_i} \prod_{i=1}^N D_i \quad \left( A_Q^{\text{total}} = \min_{\{A_i\}_i} \prod_{i=1}^N A_i \right), \end{equation} 
under the condition that $\mathcal{S_Q} = \mathcal{S}$. \end{definition}

For any given facet inequality \eqref{facetd} [or \eqref{facetan}], the classical quantities $D_C^{\text{total}}$ (or $A_C^{\text{total}}$) can be obtained straightforwardly by minimizing $\prod_i D_i$ (or $\prod_i A_i$) under the linear constraints $f_C({D_i}) = \mathcal{S}$ [or $f_C({A_i}) = \mathcal{S}$]. However, the quantum counterparts are more difficult to evaluate. Nevertheless, the semidefinite optimization method described in the previous subsection can be employed to obtain an upper bound on $D_Q^{\text{total}}$ (or $A_Q^{\text{total}}$).

In the specific case of two senders with distinguishability constraints and no input for the receiver, the following optimization using the "SeeSaw" method can be performed: \begin{align} 
\min \{\tr(\sigma_1) + \tr(\sigma_2)\} \nonumber \\ 
\text{such that } \sum_{x_{1},x_{2},z} c_{x_1,x_2,z} \Tr\left[(\rho_{x_{1}} \otimes \rho_{x_{2}}) M_{z}\right] = \mathcal{S}, \nonumber\\
~\forall i\in{1,2}, ~~\rho_{x_i}\geq 0,\Tr(\rho_{x_i})=1, \sigma_i\geq q_{x_i}\rho_{x_i};\nonumber \\
M_z\geq 0 ~\forall z, \sum_z M_z = \mathbb{I}. 
\end{align}

From the outcome of this optimization, one can deduce that the target value $\mathcal{S}$ can be achieved with $D_1 D_2 = \tr(\sigma_1)\tr(\sigma_2)$, which provides an upper bound on $D_Q^{\text{total}}$. We denote this value by $D_{Q_d}^{\text{total}}$ when $d$-dimensional quantum systems are used in the optimization. A similar approach can be employed to obtain an upper bound on $A_Q^{\text{total}}$.

The quantum advantage can be quantified by the ratio $D_C^{\text{total}}/D_Q^{\text{total}}$. When this ratio is strictly greater than 1, it implies that more information about the senders’ inputs must be communicated classically to achieve the same value of $\mathcal{S}$. Alternatively, computing $D_C^{\text{total}}/D_{Q_d}^{\text{total}}$ provides a lower bound on $D_C^{\text{total}}/D_Q^{\text{total}}$.

It is important to note that an advantage observed in one approach implies an advantage in the other, and vice versa. Specifically, the condition $\mathcal{S_Q}/\mathcal{S_C} > 1$ is equivalent to $D_C^{\text{total}}/D_Q^{\text{total}} > 1$. However, the precise values of these two ratios need not exhibit the same behavior. For example, given two different expressions for the figure of merit, one may yield a greater quantum advantage as measured by $\mathcal{S_Q}/\mathcal{S_C}$, while the other may exhibit a larger advantage according to $D_C^{\text{total}}/D_Q^{\text{total}}$.

\section{Explicit study of elementary scenarios}\label{sec3}
We focus on a two-sender and single receiver communication scenario where the first and second sender choose from $n_{x_1}, n_{x_2}$ possible inputs, respectively. The distinguishability (or antidistinguishability) of their inputs is upper bounded by $D_1$ and $D_2$ (or $A_1$ and $A_2$). We consider uniform probability distribution over inputs, i.e., $q_{x_i}=\frac{1}{n_{x_i}},\forall i\in\{1,2\}$. The receiver provides an outcome $z\in[n_z]$ depending on the received messages and a fixed input $y=1$. As the external input to the receiver is fixed, we drop the label $y$. We represent this communication scenario as $\left(n_{x_1},n_{x_2},n_z\right)$ scenario. In the following scenarios, we describe a complete characterization of the boundaries of a classical set $\mathbb{D}_C^+$ and $\mathbb{A}_C^+$ by enlisting their facet inequalities. The facet inequalities have been generated using \texttt{"polymake"} and \texttt{"julia"} software. We mention only the nontrivial facets and ignore the trivial ones of the form $0 \leqslant p\left(z|x_1,x_2\right)\leqslant 1;~ \frac{1}{n_{x_i}}\leqslant D_i\leqslant 1,1-\frac{1}{n_{x_i}}\leqslant A_i\leqslant 1 ,~\forall i\in\{1,2\}$. The term "orbit size" represents the number of inequalities that are equivalent under enlisted symmetry operations. We represent only one inequality from each equivalent class. Observe the generality of inequalities as it treats $D_i$s (or $A_i$s) as variables in respective scenarios and encompasses all distinguishability (or antidistinguishability) constrained scenarios for any valid specific values of these variables. We study lower bounds of the quantum values of these facet inequalities in the range $\frac{1}{n_{x_i}}\leqslant D_i\leqslant 1$ (or $1-\frac{1}{n_{x_i}}\leqslant A_i\leqslant 1$), using the algorithm developed in Sec. \ref{SDP}. We enlist quantum advantages in both methods as described in Sec. \ref{sub quant adv}. We study a total five scenarios constrained by either distinguishability or antidistinguishability of the senders' input. \\

 The first two scenarios $\left(2,2,2\right) $ and $\left(2,2,3\right)$ do not produce any quantum advantage. Note that the constraint of distinguishability and antidistinguishability appears to be same in these two cases as we are dealing with two inputs for both the senders. The $\left(3,2,2\right) $ scenario under the antidistinguishability constraint also does not produce any quantum advantage. We encourage the reader to look into at  Appendix \ref{no_adv} for detailed numerical foundings in these three scenarios. \\

The first instance of multipartite quantum advantage appears in $\left(3,2,2\right)$ scenario with distinguishability constraint. We provide explicit extremal communication strategies for both the senders and receiver, specifically for this scenario \cite{ankush238_preparation_measurement_vertices_2025}. We discuss one such strategy in-depth in Appendix \ref{app_3}.  We encourage the reader to see Appendix \ref{app_2} for an extensive discussion on the facet inequalities of $\mathbb{D}_C^+$ and corresponding wide range of quantum advantages. For instance, consider the following inequality from Table \ref{table3}:

\bea
\mathcal{I}_1 = p(2|1,1)-3p(2|2,1)+p(2|3,1)-p(2|1,2) \nonumber \\  +p(2|2,2)+p(2|3,2) \leqslant 6D_1+2D_2-3 .\nonumber
\eea
\begin{figure}[htbp]
    \centering    \includegraphics[scale=0.6]{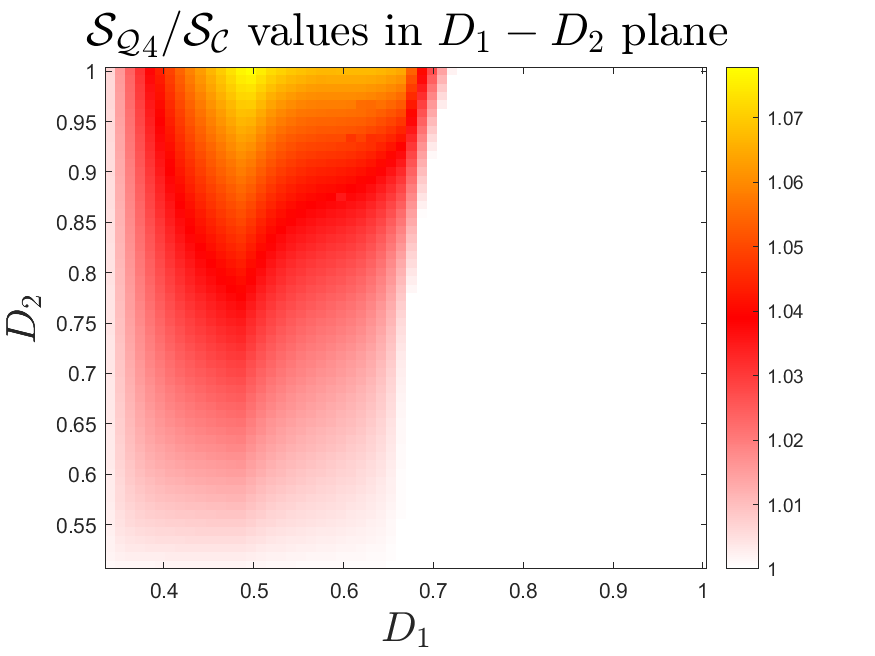}
    \caption{Landscape of $\mathcal{S}_{\mathcal{Q}_4}/\mathcal{S_C}$ over all valid choices of $D_1\in [1/3,1]$ and $D_2\in [1/2,1]$, corresponding to $\mathcal{I}_1$. }
    \label{qadvantage_I1}
\end{figure}
Figure \eqref{qadvantage_I1} depicts the obtainable quantum advantage in terms of the ratio $\mathcal{S_Q}_4/\mathcal{S_C}$ corresponding to $\mathcal{I}_1$. Following, we discuss an explicit qubit strategy of advantageous quantum communication for $\mathcal{I}_1$. Let the first sender transmit one among the three states $\rho_1=\ket{0}\bra{0},\rho_2=\ket{1}\bra{1}, \rho_3=\ket{-}\bra{-}$, 
and the second sender transmit either $\rho_1$ or $\rho_2$, while the receiver performs a binary outcome measurement \begin{align}
    M = \{ \mathbb{I}-\ket{\psi_1}\bra{\psi_1}-\ket{\psi_2}\bra{\psi_2}, ~~\ket{\psi_1}\bra{\psi_1}+\ket{\psi_2}\bra{\psi_2}\}.
\end{align}
where 
\begin{align}\ket{\psi_1}=&[0, -1/\sqrt{19}, 0, \sqrt{18}/\sqrt{19}]^T, \nonumber \\ \text{and}~ \ket{\psi_2}=&[ -\sqrt{67}/\sqrt{68}, 0,  1/\sqrt{68}, 0]^T.
\end{align} 
This strategy produces a value of $\mathcal{S_Q}=3.1796$ for $\mathcal{I}_1$, with $D_1=2/3 ,D_2=1 $. Using the same $D_1$ and $D_2$, we get $\mathcal{S_C}= 3$, which reveals an advantage of $\mathcal{S_Q}/\mathcal{S_C}=1.06.$ This is also shown in Fig. \eqref{qadvantage_I1}.



We also study the $(2,2,4)$ scenario under antidistinguishability constraints and report several pieces of evidence of quantum advantage, mentioned explicitly in Appendix \ref{app_2}. Table \ref{table4} will inscribe three of these, and we will now discuss one of them explicitly.
Consider the inequality ,
\bea
\mathcal{I}_6 =& p(1|1,2)-p(1|2,2)-p(2|1,1)+p(2|2,1)\nonumber \\
&+p(2|1,2)-p(2|2,2)+p(3|2,1)-p(3|2,2) \nonumber \\ &\leqslant  2A_1+2A_2-2 .~~~~~\nonumber
\eea
In the following, we provide an explicit example of advantageous quantum strategy for $\mathcal{I}_6$. 
Consider the following qubit strategy, where each sender transmits one of two states, $\rho_{1}=\ket{0}\!\bra{0},\rho_{2}= \ket{+}\!\bra{+}$, and the receiver performs a four-outcome measurement defined by
\begin{align} 
M=&\{\ket{\psi_1}\!\bra{\psi_1},\ket{\psi_2}\!\bra{\psi_2},\ket{\psi_3}\!\bra{\psi_3},\nonumber \\ 
&\mathbb{I}-\ket{\psi_1}\!\bra{\psi_1}-\ket{\psi_2}\!\bra{\psi_2}-\ket{\psi_3}\!\bra{\psi_3}\}, 
\end{align} 
where,
\begin{align*} \ket{\psi_1} &=[1/\sqrt{2},1/2,-1/2,0]^T, \\
\ket{\psi_2} &=[0, 1/\sqrt{3},1/\sqrt{3},-1/\sqrt{3}]^T,\\
\text{and}~ \ket{\psi_3} &=[1/\sqrt{2},-1/2,1/2,0]^T. \end{align*} 
This quantum strategy yields a value of $\mathcal{S_Q}=1.46$ for $\mathcal{I}_6$, with $A_1=A_2=0.85$. Substituting the values of $A_1,A_2$ into the classical bound of $\mathcal{I}_6$, we find an advantage of $\mathcal{S_Q}_2 / \mathcal{S_C} = 1.04$, which almost matches the best value obtained using the SeeSaw method, as will be  shown in Fig. \eqref{224subfigure 4}. On the other hand, Fig. \eqref{anti_advantage} depicts the obtainable quantum advantage in $\mathcal{I}_6$ in terms of the ratio $A_C^{\text{total}}/A_{Q_4}^{\text{total}}$, for different values of $\mathcal{S}$.

\begin{figure}[htbp]
    \centering    \includegraphics[scale=0.6]{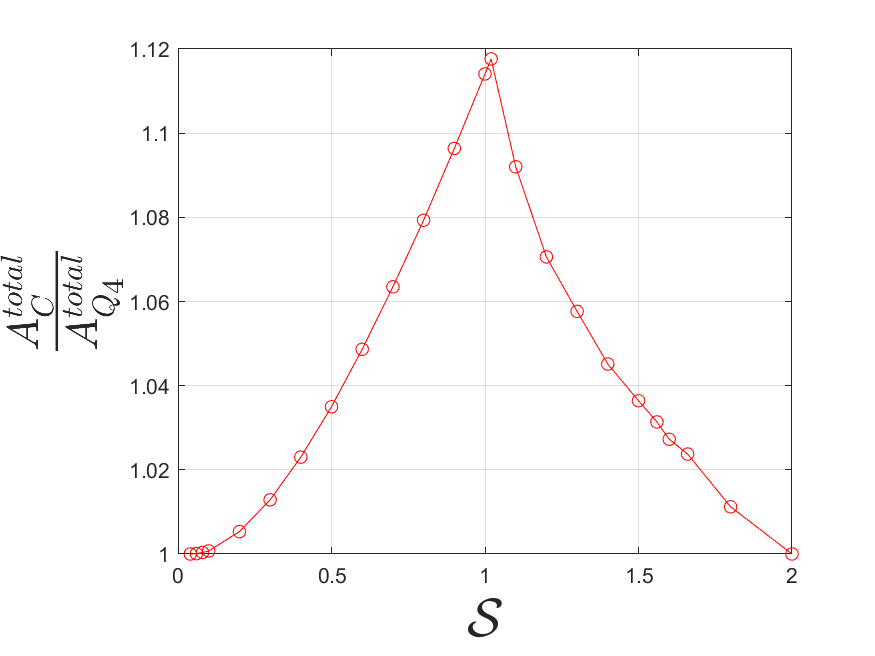}
    \caption{$A_C^{\text{total}}/A_{Q_4}^{\text{total}}$ vs $\mathcal{S}$ plot, corresponding to $\mathcal{I}_6$.}
    \label{anti_advantage}
\end{figure}

\blk

\section{The task of Antidistinguishing distributed inputs}\label{gameanti} 

Consider another multipartite communication task constrained by the antidistinguishability of the individual sender's input. Each sender $S_i$ receives an input $x_i\in[n_{x_i}]$ with probability $q_{x_i}$ and communicates a message (either classical or quantum) to the receiver depending on the input. The receiver performs a fixed measurement and obtains an outcome $z\in [\prod_i n_{x_i}]$. Let the antidistinguishability of the inputs of sender $S_i$ be bounded by $A_i$. Unlike the scenarios studied in Sec. \ref{sec3}, here we fix the success metric of the task by choosing it to be the ability of the receiver to antidistinguish the input tuples $\{\vec{x}=(x_1,...,x_N)\}$. We refer to this task as the antidistinguishability of distributed inputs defined by the success metric,
\begin{align}\label{distributed-anti}
\mathcal{\widetilde{S}} = 1-\sum_{z} q_{\vec{x}}p(z=\vec{x}|\vec{x}).    
\end{align} Here $\vec{x}$ is sampled from a probability distribution $\{q_{\vec{x}}\}_{\vec{x}}$. As each sender receives input independently from a distribution, $q_{\vec{x}}$ is essentially a product of those individual probabilities, i.e., $q_{\vec{x}} = \prod_{i=1}^Nq_{x_i}$. Note that, in this case as well, the receiver has no input. 

The success metric $\mathcal{\widetilde{S}_C}$ in classical communication can be upper bounded by a function of individual sender's antidistinguishabilities. We derive the bound in the following theorem.

\begin{thm}[Classical Antidistinguishability of distributed inputs]\label{th1}
In a classical multipartite communication scenario constrained by antidistinguishability values $\{A_i\}_i$ of the individual sender's input, the optimal success metric $(\mathcal{\widetilde{S}_C})$ of antidistinguishing distributed inputs $\vec{x}$ is upper bounded as 
    \be \label{thm1}
\mathcal{\widetilde{S}_C} \leqslant 1 - \prod_{i=1}^N (1-A_i) .
    \ee    
\end{thm}
\begin{proof}
In classical communication, the success metric \eqref{def-SC} of antidistinguishability of distributed inputs obtains the form \begin{align}\label{classical distributed}
    \mathcal{\widetilde{S}_C} = 1-\min_{\substack{\{p_e(m_i|x_i)\} \\ \{p_d(z|\vec{m},y)\}}}  \sum_{z}q_{\vec{x}}\sum_{\vec{m}} p_d\left(z=\vec{x}|\vec{m}\right)\prod_{j=1}^N p_e\left(m_j|x_j\right) ,
\end{align}
under the constraints on $\{p_e(m_i|x_i)\}$ given by \eqref{antii}.
Note that for each choice of $\vec{m}$ the decoding probabilities $\{p_d(z|\vec{m})\}$ satisfying \eqref{de1} and \eqref{de2} form a convex polytope whose extremal points are characterized by deterministic decoding strategies, i.e., $p_d\left(z|\vec{m}\right) =0~\forall z$ except for a specific outcome $z^*$ such that $p_d\left(z^*|\vec{m}\right) =1$. Corresponding to each $\vec{m}$ opting for a deterministic decoding strategy, the optimization over $\vec{x}$ can be realized by minimizing $q_{\vec{x}}\prod_{j=1}^Np_e\left(m_j|x_j\right)$ over the index $\vec{x}$. This leads to a simpler reformulation of \eqref{classical distributed} as \begin{equation}  \label{stc-sim}
\mathcal{\widetilde{S}_C} =1-\min_{\{p_e(m_i|x_i)\}} \sum_{\vec{m}}\min_{\vec{x}}\left\{\prod_{j=1}^N q_{x_j} p_e\left(m_j|x_j\right) \right\} ,
\end{equation}
where $q_{\vec{x}}=\prod_{j=1}^Nq_{x_j}$ is used. It can be easily seen that for any set of non-negative numbers $\{\alpha_{i,j}\}_{i,j},$ the following identity holds,
\begin{equation}
    \min_i \left\{ \prod_{j} \alpha_{i,j} \right\} = \prod_j \min_i \left\{ \alpha_{i,j} \right\}.
\end{equation}
By identifying the variables $\alpha_{i,j}$ with $q_{x_j}p_e(m_j|x_j)$ in \eqref{stc-sim}, and applying the above relation, we obtain
\begin{align}
\mathcal{\widetilde{S}_C} &=1- \min_{\{p_e(m_i|x_i)\}} \sum_{\vec{m}}\prod_{j=1}^N\min_{x_j} \left\{ q_{x_j}p_e\left(m_j|x_j\right) \right\} \nonumber\\ &= 1-\min_{\{p_e(m_i|x_i)\}} \prod_{j=1}^N \sum_{m_j} \min_{x_j} \left\{q_{x_j} p_e\left(m_j|x_j\right) \right\} \nonumber \\ &\leq 1-\prod_{i=1}^N\left(1-A_i\right) .
    \end{align} 
Here the second line is only a rearrangement of the previous line. The third line follows from the definition of $A_i$ in \eqref{antii}. This completes the proof.
\end{proof}

We further constrain this task of antidistinguishing distributed inputs by considering $ A_i$ to be the same for all the senders. Now, consider $A^C$ to be the required antidistinguishability of each sender's inputs in classical communication to achieve a specific value of success metric $\mathcal{\widetilde{S}} = \mathcal{\widetilde{S}_C}$. Analogously, let $A^Q$ be the required antidistinguishability of individual sender's inputs in quantum communication to achieve the same value $\mathcal{\widetilde{S}} = \mathcal{\widetilde{S}_Q}$. We consider the ratio $(A^C/A^Q)^N$ as the quantifier of quantum advantage for this task. 
 
In the following theorem, we prove the existence of a multipartite quantum communication protocol that is advantageous over classical communication. For simplicity, we consider equiprobable inputs for all the senders, i.e., $q_{x_i}=1/n_{x_i}~ \forall i\in [N]$.

\begin{thm}[Quantum advantage in antidistinguishing distributed inputs inspired by PBR Theorem \cite{Pusey_2012}]\label{pbr measurement} 
There exists a quantum multipartite communication protocol constrained by antidistinguishability, having exponential advantage with the number of senders $(N)$ when $\mathcal{\widetilde{S}}=1$,
\begin{align} \label{exp-adv}
    \left(\frac{A^C}{A^Q}\right)^N =  \left(\frac{2}{1+\sin\theta}\right)^N>1
\end{align} where $\theta$ satisfies $2\tan^{-1}(2^{(1/N)}-1)\leq \theta <\pi/2$.
\end{thm}
\begin{proof}
Consider each sender receives two equiprobable inputs $x_i\in\{0,1\}~\forall i\in [N]$ and communicates a quantum state $\ket{\psi_{x_i}}=\cos(\theta/2)\ket{0}+(-1)^{x_i}\sin(\theta/2)\ket{1}$ to the receiver, where $x_i\in \{0,1\}$. The receiver gets a product state $\ket{\psi_{\vec{x}}}=\ket{\psi_{x_1}}\otimes\cdots\otimes\ket{\psi_{x_N}}$. The task of the receiver is to antidistinguish these $2^N$ possible product states $\{\ket{\psi_{\vec{x}}}\}_{\vec{x}}$. Reference \cite{Pusey_2012} proved the existence of a $2^N$-outcome joint measurement on $N$ qubit state such that each outcome has zero probability of occurrence for one of the inputs when $\theta$ satisfies \begin{align}\label{theta_bound}
    2\tan^{-1}(2^{\frac{1}{N}}-1)\leq \theta <\pi/2 .\end{align} Essentially, this measurement strategy successfully antidistinguishes these product states, i.e., $\min_{\vec{x}}\Tr(\rho_{\vec{x}}M_{z})=0 ~\forall z$, which implies $\mathcal{\widetilde{S}_Q} =1$. The required antidistinguishability of individual sender's input to obtain $\mathcal{\widetilde{S}_Q} =1$, can be calculated using the Hellstrom formula \cite{Helstrom1969}. Antidistinguishability of the two input states $\ket{\psi_0}$ and $\ket{\psi_1}$ is, \begin{align}
        A^Q = & \frac{1}{2} \left( 1+\sqrt{1-|\langle\psi_0|\psi_1\rangle|^2}\right)\nonumber\\
  =&\frac{1}{2} \left(1+\sin\theta\right).
    \label{sin theta}\end{align} To achieve the same value of antidistinguishability of distributed inputs in classical communication, i.e., for $\mathcal{\widetilde{S}_C}=\mathcal{\widetilde{S}_Q}=1$, $A^C$ should be $1$, following from \eqref{thm1}. All the senders opt for the same communication strategy of sending a classical bit $i$ when its input is $i$, where $i\in \{0,1\}$. These two classical preparations are perfectly antidistinguishable, i.e., $A^C=1$. Consequently, we have \eqref{exp-adv} for any choice of $\theta$ from the range defined in \eqref{theta_bound}. This completes the proof.
\end{proof}
Although Theorem \ref{pbr measurement} guarantees the existence of a quantum protocol that is advantageous over classical communication, the amount of quantum advantage may vary drastically depending on the choice of $\theta$. Choosing $\theta \approx \pi/2$ from below, i.e., $\ket{\psi_{x_i}}\approx\{\ket{+},\ket{-}\}$, leads to $(A^C/A^Q)^N\approx 1$, even if the number of senders $N$ is large. So for this choice of $\theta$, despite of using a large number of senders, the quantum communication protocol fails to produce a large advantage over classical communication, i.e., essentially the exponential advantage is lost. 
This signifies the choice of $\theta$ is crucial to obtain a large quantum advantage despite using less senders.

This motivates the following corollary.
\begin{coro} The optimum quantum advantage to antidistinguish distributed inputs
under the restriction that the senders and receivers follow the quantum strategy based on
the PBR  \cite{Pusey_2012} construction, is achieved when $\theta = 2\tan^{-1}(2^{\frac{1}{N}}-1)$ for finite $N$ and the quantum advantage turns out to be 
\begin{align} \label{prb-best-ratio}
    \left(\frac{A^C}{A^Q}\right)^N =  2^N\left(1+2^{1-\frac{2}{N}}-2^{1-\frac{1}{N}}\right)^N .
\end{align} 
In the case when $N\rightarrow\infty$, this leads to 
\begin{align}
    \left(\frac{A^C}{A^Q}\right)^N\rightarrow 2^N .
\end{align}
\end{coro}
\begin{proof}
    As $\sin\theta$ is a monotonically increasing function in the range of $\theta$ defined in \eqref{theta_bound}, the optimal choice of $\theta$ to maximize $2/(1+\sin\theta)$ corresponds to the minimum allowed value of $\theta$. This leads to the choice of optimal initial $N$-qubit product state for the quantum circuit used in \cite{Pusey_2012}. The minimum value of $\sin\theta$ is achieved when $\tan(\theta/2) =2^{\frac{1}{N}}-1$. Using trigonometric relations we get \begin{align}\sin\theta &= \frac{2\left(2^{\frac{1}{N}}-1\right)}{1+\left(2^{\frac{1}{N}}-1\right)^2} .
    \end{align} 
    After substituting this expression in the right-hand side of \eqref{exp-adv}, the quantum advantage turns out to be \eqref{prb-best-ratio}. This completes the proof.
\end{proof}

Until now, our discussion has been limited to scenarios where each sender has access to binary input choices. We now consider a general scenario further and allow an arbitrary $n$ input choices for each sender. Naturally, the question arises whether there exists any quantum strategy that is advantageous over classical communication for this task. In the following theorem, we provide an affirmative answer to this question. 
\begin{thm}[Sufficient condition for quantum advantage] Consider each sender gets $n$ possible inputs, and communicates the same set of states $\{\ket{\psi_{x}}\}_x$ to the receiver. A quantum advantage with $\mathcal{\widetilde{S}_Q}=1$ arises whenever there exists a set of states $\{\ket{\psi_{x}}\}_x$ whose inner products satisfy the following relations:
\begin{align}\label{suff-cond-1}
    n(n-2)< \sum_{j,l;j\neq l}|\langle\psi_j|\psi_l\rangle|, ~~~j,l\in[n]
\end{align} and 
\begin{align} \label{suff-cond-2}
    \sum_{j,l;j\neq l}|\langle\psi_j|\psi_l\rangle|^2\leq \frac{n^2}{2^{1/N}}-n  .
\end{align} 
\end{thm}
\begin{proof}
In order to gain quantum advantage, i.e., $\left(\frac{A^C}{A^Q}\right)^N>1$, when $\mathcal{\widetilde{S}_C}=\mathcal{\widetilde{S}_Q}=1$, it requires $A^Q<1$ since $A^C$ must be $1$ due to Theorem \ref{th1}. This implies the states $\{\ket{\psi_j}\}_{j=1}^n$ should not be perfectly antidistinguishable. In this regard, there already exists a necessary condition \eqref{suff-cond-1} [Theorem $4.1$ of \cite{johnston2023tight}] for a set of $n$ states to be not antidistinguishable. 

On the other hand, $\mathcal{\widetilde{S}_Q}=1$ necessitates that the $n^N$ number of product states $\{\ket{\psi_{\vec{x}}}\}_{\vec{x}}$, where each  $\ket{\psi_{\vec{x}}}=\ket{\psi_{x_1}}\otimes\cdots\otimes\ket{\psi_{x_N}}$, be antidistinguishable at the receiver's end. In order to antidistinguish this set of $n^N$ states, it suffices to show that the Gram matrix $(G)$, constructed from these states, satisfies the following relation (Corollary $5.3$ of \cite{johnston2023tight}): \begin{align}\label{gram_matrix}
        ||G||_F\leq \frac{n^N}{\sqrt{2}}.
    \end{align} Here $||G||_F$ refers to the Frobenius norm of the Gram matrix associated with the states $\{\ket{\psi_{\vec{x}}}\}_{\vec{x}}$. Let, \begin{align}\sum_{j,l;j\neq l}|\langle\psi_{j}\ket{\psi_{l}}|^2=\alpha .\end{align} Using the definition of Gram matrix, we can write \begin{align}
        ||G||_F^2 &= \sum_{\vec{x},\vec{x'}}|\langle\psi_{\vec{x}}\ket{\psi_{\vec{x'}}}|^2 \nonumber\\ &= \sum_{k=0}^N\sum_{\vec{x}_k,\vec{x'}_k} |\langle\psi_{\vec{x}_k}\ket{\psi_{\vec{x'}_k}}|^2 \nonumber\\ &=\sum_{k=0}^N \prescript{N}{}{C}_k n^{N-k}\alpha^k \nonumber\\ &=(n+\alpha)^N \label{gram_matrix_square}.
    \end{align} 
    Here $k$ is the number of positions where $\vec{x}_k$ and $\vec{x'}_k$ differ from each other. $\vec{x}_k$ and $\vec{x'}_k$ belongs to the set of $N$-tuples $\{\vec{x}\}$. The third line follows from the fact that there are $\prescript{N}{}{C}_k$ possible combinations in which $\vec{x}_k$ and $\vec{x'}_k$ can differ. The rest of the $N-k$ entries can be the same in $n^{N-k}$ possible ways. Finally, taking all possible combinations where $k$ positions vary, the sum of the squared inner products becomes $\alpha^k$. Combining \eqref{gram_matrix} and \eqref{gram_matrix_square} we get \begin{align}
        (n+\alpha)^{N/2} &\leq \frac{n^N}{\sqrt{2}} \nonumber \end{align} which implies,\begin{align} \alpha &\leq \frac{n^2}{2^{1/N}}-n.
    \end{align} This completes the proof.
\end{proof}    

In the following, we provide explicit examples of quantum advantage for this communication task, where each sender has more than two inputs.

Consider the multipartite scenario comprising two senders, and each sender has three possible inputs. Based on input $x_i\in[3]$, the senders communicate $\{\ket{\psi_{x_i}}\}$ to the receiver, where 
\begin{eqnarray}\label{states2,3}
\ket{\psi_1} &=& \ket{0},\nonumber \\
\ket{\psi_2} &=& \cos\Bigg({\frac{5\pi}{18}}\Bigg)\ket{0} + \sin\Bigg({\frac{5\pi}{18}}\Bigg)\ket{1}, \nonumber \\
\ket{\psi_3} &=& \cos\Bigg({\frac{19\pi}{60}}\Bigg)\ket{0} + e^{\frac{i2\pi}{3}}\sin\Bigg({\frac{19\pi}{60}}\Bigg)\ket{1}.
\end{eqnarray}
It can be checked that these states are not antidistinguishable, using the necessary and sufficient conditions described in \cite{PhysRevA.66.062111}. A semidefinite program yields the value of antidistinguishability of these states to be $0.9798$.
The receiver can get nine different bipartite product states, i.e., $\ket{\psi_j}\otimes\ket{\psi_l}, ~\forall j,l\in[3]$.

A semidefinite program yields the antidistinguishability of these product states to be $1$. This implies $\mathcal{\widetilde{S}}_Q=1$. On the other hand, due to Theorem \ref{th1}, $\mathcal{\widetilde{S}_C}=1$ is obtained in classical communication when $A^C =1$. This results in a quantum advantage of the amount $(A^C/A^Q)^2 = 1/0.9798^2= 1.042$.

We further study three different multipartite communication scenarios for this task using SDP. The three scenarios, labeled by the total number of senders $N$ and the number of inputs $n$ of each sender, are $(N=2;n=3),(N=2; n =4)$ and $(N=3;n=3)$. Implementing the SeeSaw method, we observe that the respective $\mathcal{\widetilde{S}_Q}$ is maximum when all the senders share the same value of antidistinguishability for their inputs. This observation further motivates us to consider quantum strategies in which all senders share the same value of antidistinguishability. Considering $A_i=A^C$ for all $i$ in the expression of \eqref{thm1}, we get
\begin{align}
    A^C \geq 1-\left(1-\mathcal{\widetilde{S}_C}\right)^{\frac{1}{N}},
\end{align} 
where the same value of the success metric in classical communication is taken, that is, $\mathcal{\widetilde{S}}=\mathcal{\widetilde{S}_C}=\mathcal{\widetilde{S}_Q}$.
For different values of the success metric $\mathcal{\widetilde{S}}$, we study quantum advantage via the ratio $(A^C/A^Q)^N$ and plot the obtained values in Fig. \ref{game_antidis}. Among the three scenarios $(N=3;n=3)$ achieves highest quantum advantage as $(A^C/A^Q)^N \approx 1.7489$ for $\mathcal{\widetilde{S}}=1$.

\begin{figure}[h!]
    \centering
    \includegraphics[scale=0.45]{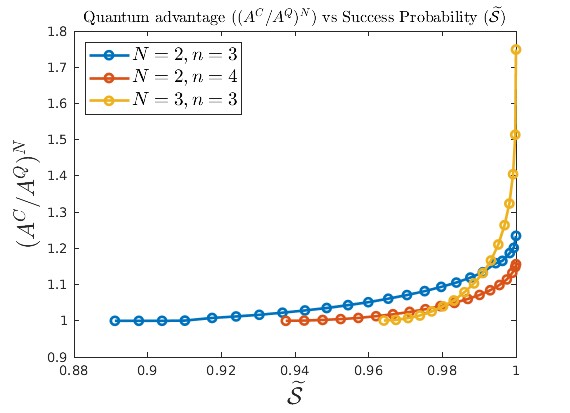}
    \caption{Quantum advantage vs success probability for the task of antidistinguishing distributed inputs with qubit states. Here $(A^C/A^Q)^N> 1$ signifies quantum advantage. } 
    \label{game_antidis}
\end{figure}

\section{Quantum Advantage and Epistemic incompleteness}\label{epistemic incompleteness}

A prepare-and-measure scenario involves a set of preparation and measurement procedures where preparation $(x)$ of a physical system is followed by a measurement $(y)$, resulting in an outcome $(z)$. An operational theory merely describes the probabilities $\{p(z|x,y)\}$ of a prepare-and-measure experiment. An ontological model of an operational theory attempts to offer an explanation of the outcome probabilities $p(z|x,y)$ assuming the existence of \textit{states of reality} also known as \textit{ontic states} $(\lambda)$. Ontic states describe the real state of affairs of the physical system irrespective of the later being subject to a measurement or not. In an ontological model each preparation corresponds to a probability distribution $\{\mu(\lambda|x)\}_{\lambda}$ over the space of ontic states $\Lambda :=\{\lambda\}$. Each measurement $y$ with outcome $z$ corresponds to a set of response functions $\{\xi(z|\lambda,y)\}_{\lambda}$. The operational probabilities are reproduced by averaging over the knowledge of $\lambda$: 
\begin{equation}
    p(z|x,y) = \int_{\Lambda} \mu(\lambda|x) \xi(z|y,\lambda) ~d\lambda.
\end{equation}
With the analog of multi-sender communication, here each preparation $\vec{x}=(x_1,...,x_N)$, of the combined physical system corresponds to a distribution $\mu(\vec{\lambda}|\vec{x})$ over the ontic space $\Lambda$, where $\vec{\lambda}=(\lambda_1,...,\lambda_N)$. As we consider all the senders to be independent, the assumption of \textit{preparation independence} \cite{Pusey_2012} asserts that
\begin{equation}
    \mu(\vec{\lambda}|\vec{x}) = \prod_{i=1}^N \mu(\lambda_i|x_i) ,
\end{equation}
along with the existence of a product space $\Lambda = \otimes_i \Lambda_i$, where $\lambda_i \in \Lambda_i$ is the ontic variable assigned to the preparation of sender $S_i$. Consequently, in such models, any operational probability can be obtained as
\begin{equation} \label{ei-pzxy}
    p(z|\vec{x},y) = \int_{\Lambda_1}d\lambda_1 \cdots \int_{\Lambda_N}d\lambda_N \left(\prod_i \mu(\lambda_i|x_i) \right) \xi(z|y,\vec{\lambda}) .
\end{equation} Any operational property of a set of preparations $\{x_i\}_{x_i=1}^{n_{x_i}}$ of $i$-th sender can be defined as 
\begin{equation}\label{obs property}
    \mathcal{P} = \max_{M} \sum_{x_i,z} c_{z,x_i} \ p(z|x_i,M) ,
\end{equation} where $\{c_{z,x_i}\in \mathbb{R}\}_{z,x_i}$ are real coefficients and the maximization is over all possible measurements. Now we define a notion of classicality in its most fundamental form. \begin{definition}[Epistemically complete model]
    An ontological model is said to be epistemically complete if the underlying epistemic states can fully explain the observed value $\mathcal{P}$ [defined in \eqref{obs property}] of a given operational property \cite{epi_incom}, essentially saying \begin{equation}
  \int_{\Lambda_i} \max_{z} \left\{ \sum_{x_i} c_{z,x_i} \mu(\lambda_i|x_i) \right\} ~d\lambda_i = \mathcal{P}. 
\end{equation}
\end{definition}
Consequently, an operational theory is said to be \textit{epistemically incomplete} if the epistemic states fail to reproduce the value $\mathcal{P}$ [defined in Eq. \eqref{obs property}] of an observational property.
Now, if we consider the operational property of a set of preparations $\{x_i\}$ to be the distinguishability $D_i$, then an epistemically complete model reproduces the observed value as
\begin{equation} \label{ei-D}
    \int_{\Lambda_i} \max_{x_i} \left\{ q_{x_i} \mu(\lambda_i|x_i) \right\}~d\lambda_i = D_i.  
\end{equation}
This notion of classicality \eqref{ei-D} was introduced in \cite{bod} by the name of \textit{bounded ontological distinctness}. Similarly, if we take the operational property of a set of preparations $\{x_i\}$ to be the antidistinguishability $A_i$, then an epistemically complete model reproduces the observed value as 
\begin{equation}\label{ei-A}
 1 -  \int_{\Lambda_i} \min_{x_i} \left\{ q_{x_i} 
 \mu(\lambda_i|x_i) \right\}~d\lambda_i = A_i. 
\end{equation}
Such an ontological model studied in \cite{ray2024epistemicmodelexplainantidistinguishability}, is analogous to the notion of \textit{maximally epistemic model} concerning the common overlap of a set of preparations. Therefore, in general, any successful epistemically complete ontological model satisfying preparation independence should reproduce the operational probabilities \eqref{ei-pzxy} in a communication scenario constrained either by distinguishability or antidistinguishability by satisfying the conditions, either \eqref{ei-D} or \eqref{ei-A} respectively, for all $i\in [N]$. This is equivalent to multipartite classical communication, wherein the operational probability is obtained by Eq.~\eqref{pcl} under the communication constraints of either \eqref{dis} or \eqref{antii}, for distinguishability -or antidistinguishability-constrained scenario, respectively. Based on this discussion, we arrive at the following observation. \begin{obs}
    Any advantage in a multipartite communication scenario under distinguishability or antidistinguishability constraint implies epistemic incompleteness under the assumption of preparation independence.  
\end{obs} As demonstrated earlier in Sec.\ref{sec3}, quantum advantage in multipartite communication scenarios studied in this article implies epistemic incompleteness of quantum theory under the preparation independence assumption.

It is worth mentioning that Spekkens' notion of contextuality \cite{PhysRevA.71.052108,PhysRevLett.102.010401} in prepare-and-measure scenarios can be viewed as a special case of a distinguishability-constrained communication task. Preparation noncontextuality requires that the epistemic states associated with two operationally equivalent preparations, say $x$ and $x'$ must be identical. Formally, if for all $z,M,$ $p(z|x,M)=p(z|x',M)$, then preparation noncontextuality demands $\forall \lambda$, $\mu(\lambda|x)=\mu(\lambda|x')$. Now consider Eq. \eqref{ei-D}. Whenever two preparations for any sender are completely indistinguishable, that is, operationally equivalent, the distinguishability $D$ becomes $\max\{q_x,q_{x'}\}$. In turn, the relation \eqref{ei-D} holds only if $\forall \lambda,$ $p(\lambda|x) = p(\lambda|x')$, which is exactly the preparation-noncontextuality condition.

\section{conclusion} 

In this article we have investigated multipartite communication involving multiple senders and one receiver, under the constraints imposed by either the distinguishability or the antidistinguishability of each sender's inputs. Most importantly, the receiver is not fed with any external inputs, unlike the RAC scenario studied in \cite{PRR_2024}. As the constraints are independent of the dimensionality of the communicated systems, the observed quantum advantage in communication tasks captures a fundamental feature of multipartite communication, having no analog to single sender and single receiver tasks. Comparison between distinguishability- and dimensionality-constrained approaches is discussed in \cite{PRR_2024}. Multiple-copy state discrimination problem can be thought of as a special case of this task, assuming each sender prepares the same state corresponding to similar input indices. On a separate note, contextuality scenarios can be thought of as a special case of distinguishability based approaches \cite{bod}. 

To provide a complete characterization of the set of classical correlations, we develop a method to derive the facet inequalities that define its boundary. These inequalities are explicitly computed in elementary scenarios involving two senders and no input at the receiver. We have performed a comprehensive study of quantum correlations using semidefinite programming to identify violations of classical bounds. 

Further, in order to demonstrate an unbounded quantum advantage, we have examined the task of antidistinguishing the inputs distributed across all senders. First, we have derived upper bounds on the success metric of this task in classical communication. Using the well-known Pusey-Barrett-Rudolph theorem, we have shown that when each sender has two possible inputs, the quantum advantage increases exponentially with the number of senders. Additionally, we have established the correspondence between quantum advantage in such multipartite scenarios and the foundational notion of epistemic incompleteness.

Our results open several promising directions for future research, some of which are listed here. While we have demonstrated exponentially increasing quantum advantage with the number of senders, the possibility of obtaining such an advantage with a finite number of senders and no input on the receiver remains an open question. Extending our analysis to more general communication networks can reveal more interesting features of quantum communication. The role of shared entanglement can be studied in enhancing classical communication protocols within our framework, which may offer deeper insights into the interplay between quantum resources in classical communication and quantum advantage.  Deriving criteria, on which facet inequalities are most likely to witness any quantum advantage, may reveal deeper insights into the relation between quantum advantage and the geometry of classical correlation sets. 
The relation of quantum advantage in contextual scenarios with the epistemic incompleteness is an interesting future direction \cite{shahandeh}. Finally, investigating quantum advantage in more practically motivated communication scenarios may help bridge the gap between foundational research and real-world quantum technologies.

\subsection*{Acknowledgment}
This work is supported by STARS (STARS/STARS-2/2023-0809), Government
of India. SH acknowledges funding from the Ministry of Electronics and Information Technology (MeitY), Government of India, under Grant No. 4(3)/2024-ITEA. AP thanks UGC, India for Junior Research Fellowship.

\bibliography{ref}

\appendix
\onecolumngrid
\section{No advantage scenarios}\label{no_adv}

\subsection{$\left(2,2,2\right) $ scenario}

This is the simplest scenario in multipartite communication. For two choices of input of each sender, the distinguishability and antidistinguishability constraints (\ref{op3}) and (\ref{anti2}) are essentially the same, effectively saying $D_1 = A_1, D_2 = A_2$ and $\mathbb{D}_C^+ = \mathbb{A}_C^+$.
We obtained $18$ inequalities out of which $10$ are trivial, and the rest can be classified into a single equivalence class upon applying symmetry conditions. The following inequality represents the rest of the eight nontrivial inequalities 

 \begin{align}
     p(1|2,1)-p(1|2,2) \leqslant 2D_{2}-1.
 \nonumber\end{align}
    The symmetry conditions applied here are as follows : $1\xlongleftrightarrow {x_1} 2, 1\xlongleftrightarrow{x_2} 2, 1\xlongleftrightarrow {z} 2 $ and $(x_1\xlongleftrightarrow{}x_2) \wedge (D_1\xlongleftrightarrow{}D_2)$, where  $j\xlongleftrightarrow {x_i} k$ indicates equivalence between two possible choices  
$j,k$ of the variable $x_i$. $x_1\xlongleftrightarrow{}x_2$ represents the symmetry of exchanging the labels between two senders. The notation $\wedge$ between two symmetries implies that they are applied together. This notation is followed throughout the article. 
There is no quantum violation of this inequality, i.e, $\mathcal{S_Q}= \mathcal{S_C}$ and $\mathbb{D}_C =\mathbb{D}_Q =\mathbb{A}_C =\mathbb{A}_Q$ for any valid choice of the distinguishability or antidistinguishability variables.

\subsection{$\left(2,2,3\right) $ scenario}
In this scenario with bounds either on distinguishability or antidistinguishability of the sender's input,  Table \ref{223dist} shows the obtained facet inequalities of both 
 $\mathbb{D}_C^+$ and $\mathbb{A}_C^+$, essentially saying $\mathbb{D}_C^+=\mathbb{A}_C^+$. As the notion of distinguishability and antidistinguishability are equivalent for two inputs, here also we have $D_1=A_1,D_2=A_2$. 


\begin{table}[H]
    \centering
    \begin{tabular}{|c|c|}
    \hline
         Orbit size& Facet  Inequalities  \\
         \hline
         24&$-p\left(2|2,1\right)+p\left(2|2,2\right)-p\left(3|2,1\right)+p\left(3|2,2\right)\leqslant 2D_{2}-1$  \\
         \hline
         48& $-2p\left(2|1,1\right)+p\left(2|2,1\right)+2p\left(2|1,2\right)-p\left(2|2,2\right)+2p\left(3|1,2\right)-2p\left(3|2,2\right) \leqslant 4D_{1}+2D_{2}-3 $ \\
         \hline
         24&$-p\left(2|2,1\right)+p\left(2|1,2\right)-p\left(3|1,1\right)+2p\left(3|1,2\right)-p\left(3|2,2\right) \leqslant2D_{1}+2D_{2}-2$  \\
         \hline
         24& $-2p\left(2|2,1\right)+2p\left(2|1,2\right)-p\left(3|1,1\right)+p\left(3|2,1\right)+3p\left(3|1,2\right)-3p\left(3|2,2\right) \leqslant 4D_{1}+4D_{2}-4 $ \\
         \hline
         
    \end{tabular}
\caption{$134$ facet inequalities were obtained, out of which $14$ are trivial and the rest can be classified into four different equivalence classes listed above. Applied symmetry conditions are 
    $1 \xlongleftrightarrow{x_1} 2$,  $1 \xlongleftrightarrow{x_2} 2$,  $1 \xlongleftrightarrow{z} 2$, $2 \xlongleftrightarrow{z} 3$, $1 \xlongleftrightarrow{z} 3$, $(x_1\xlongleftrightarrow{}x_2) \wedge (D_1\xlongleftrightarrow{}D_2)$.}
    \label{223dist}
\end{table}

There is no quantum violation for any of the facet inequalities listed in Table \ref{223dist}. This suggests that the set of quantum correlations and classical correlations might be the same, i.e, $\mathbb{D}_C = \mathbb{D}_Q=\mathbb{A}_C = \mathbb{A}_Q$ for any valid choices of $D_1$ and $D_2$.

\subsection{$\left(3,2,2\right)$ antidistinguishability scenario }
In this scenario with bounds on antidistinguishability of sender's input, Table \ref{table 2} enlists the set of facet inequalities characterizing the classical set $\mathbb{A}_C^+$. We obtain no quantum violation in this (3,2,2) antidistinguishability scenario.
 

\begin{table}[H]
    \centering
    \begin{tabular}{|c|c|c|}
    \hline
      Orbit size   & Facet Inequalities   \\
         \hline
        12 & $p(1|1,1)-p(1|2,1) \leqslant 3A_{1}-2$ \\
        \hline
        6 & $-p(1|1,1)+p(1|1,2)\leqslant 2A_{2}-1$  \\
        \hline
        12& $-p(1|2,1)+2p(1|3,1)-2p(1|1,2)+p(1|2,2)\leqslant 6A_{1}+2A_{2}-5$  \\
        \hline
    \end{tabular}
\caption{44 facet inequalities were obtained, out of which $14$ are trivial and the rest can be classified into three different equivalence classes. Applied symmetry conditions are  $1 \xlongleftrightarrow{x_1} 2$; $2 \xlongleftrightarrow{x_1} 3$; $1 \xlongleftrightarrow{x_1} 3$ ;$1 \xlongleftrightarrow{x_2} 2$; $1 \xlongleftrightarrow{z} 2$.}
\label{table 2}
\end{table}

\section{Advantageous scenarios}\label{app_2}

\subsection{$\left(3,2,2\right)$ distinguishability scenario} 

A multipartite communication scenario with bounded distinguishability on both senders' input. Table \ref{table3} enlists all the facet inequalities of $\mathbb{D}_C^+$. The first two inequalities in Table \ref{table3} do not produce any quantum violation, i.e., $\mathcal{S_Q}=\mathcal{S_C}$. The behaviors of quantum violation of $\mathcal{I}_2,\mathcal{I}_3$ and $\mathcal{I}_4$ are depicted in Fig. \ref{Fig3}.  

\begin{table}[H]
    \centering
    \small
    \renewcommand{\arraystretch}{1.9}
    \begin{tabular}{|c|c|c|}
    \hline
          Orbit & Facet  Inequalities & Maximum obtained quantum advantage  \\
         size &  & in terms of $\frac{D_C^{\text{total}}}{D_{Q_4}^{\text{total}}}$ for fixed $\mathcal{S}$ \\
         \hline
         6 & $-p(2|3,1)+p(2|3,2) \leqslant 2D_2-1$ & No advantage\\ 
         \hline
         12 & $p(2|1,1)-p(2|2,1) \leqslant 3D_1-1$ & No advantage \\ 
         \hline
         24 &  $\mathcal{I}_1 = p(2|1,1)-3p(2|2,1)+p(2|3,1)-p(2|1,2)+p(2|2,2)+$  & $\frac{D_C^{\text{total}}}{D_{Q_4}^{\text{total}}} =1.0517$ at $\mathcal{S}=2.1339$ \\ 
         & $p(2|3,2) \leqslant 6D_1+2D_2-3$ & \\
         \hline
         24 & $\mathcal{I}_2 = 2p(2|1,1)-2p(2|2,1)+p(2|3,1)-2p(2|1,2)+p(2|3,2) \leqslant 6D_1+2D_2-3$ & $\frac{D_C^{\text{total}}}{D_{Q_4}^{\text{total}}}=1.06093$ at $\mathcal{S}=3.1579$ \\
         \hline
         12 & $\mathcal{I}_3 = p(2|1,1)-3p(2|2,1)+3p(2|3,1)+3p(2|1,2)-p(2|2,2)$ &$\frac{D_C^{\text{total}}}{D_{Q_4}^{\text{total}}}=1.04715$ at $\mathcal{S}=4.3975$ \\ 
         & $-3p(2|3,2) \leqslant 12D_1+2D_2-5$ & \\
         \hline
         24 & $\mathcal{I}_4 = -p(2|1,1)+p(2|2,1)-p(2|1,2)-p(2|2,2)+p(2|3,2) \leqslant 3D_1-1$ & $\frac{D_C^{\text{total}}}{D_{Q_4}^{\text{total}}}= 1.0121$ at $\mathcal{S}=0.0323$\\ 
         \hline
         \end{tabular}
    \caption{
    $116$ facet inequalities were obtained, among which $14$ are trivial and the rest can be grouped into six distinct equivalence classes. Applied symmetry conditions are  $1\xlongleftrightarrow{x_1}2 $;$2\xlongleftrightarrow{x_1}3 $;$1\xlongleftrightarrow{x_1}3 $; $1\xlongleftrightarrow{x_2}2 $; $1\xlongleftrightarrow{z}2 $.}
    \label{table3}
\end{table}

\begin{figure}[htpb]
    \centering
    \begin{subfigure}[b]{0.5\textwidth}
        \centering
        \includegraphics[scale=0.6]{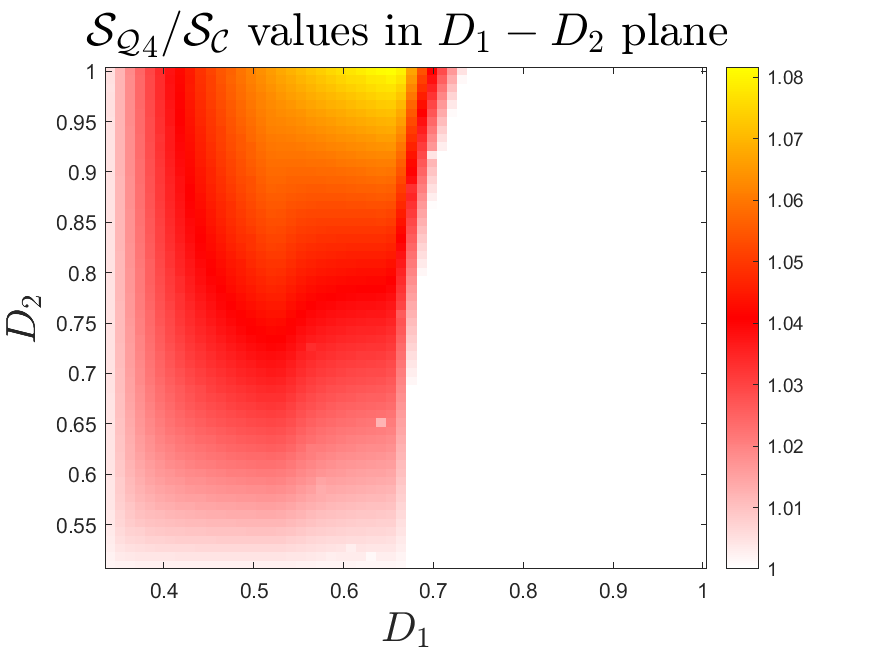}
        \caption{}
         \label{3b}
    \end{subfigure}
     \hfill
    \begin{subfigure}[b]{0.49\textwidth}
        \centering
        \includegraphics[scale=0.6]{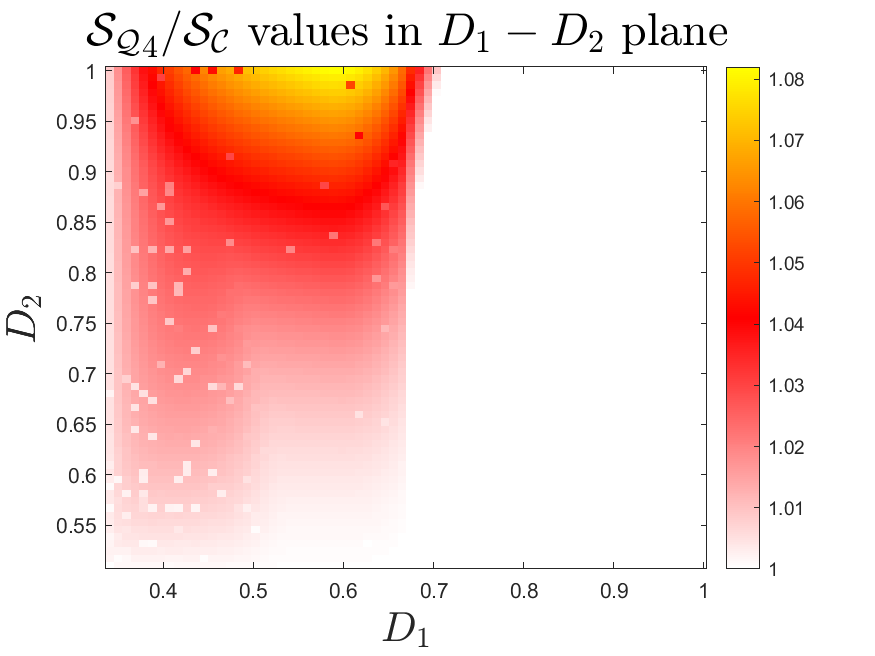}
        \caption{}
         \label{3c}
    \end{subfigure}
   
\caption{Landscape of $\mathcal{S}_{\mathcal{Q}_4}/\mathcal{S_C}$ over all valid choices of $D_1$ and $D_2$ in $(3,2,2)$ distinguishability scenario. Figures \eqref{3b} and \eqref{3c} respectively correspond to the three inequalities $\mathcal{I}_2$ and $\mathcal{I}_3$ in Table \ref{table3}. In particular, we consider values of $D_1 \in [1/3,1]$ and $D_2 \in [1/2,1]$ sampled at small intervals. For each pair $(D_1,D_2)$, we compute the left-hand side of the corresponding inequality using the aforementioned SeeSaw optimization technique applied to four-dimensional quantum systems. The classical bound is determined by substituting the same values of $D_1$ and $D_2$ into the classical bound of the inequality. We subsequently calculate the ratio of the quantum value to the classical bound and represent this ratio using a color-coded grid, where each grid point corresponds to a specific pair of $D_1$ and $D_2$ values. The blurred white dots appear due to precessional error in SeeSaw optimization.}   
    \label{Fig3}
\end{figure} 





\subsection{$\left(2,2,4\right)$ scenario }
Multipartite communication scenario, with bounds on distinguishability (or anti-distinguishability) of the senders' input. Note that, in this scenario, we have $D_1=A_1, D_2=A_2$ and $\mathbb{D}_C^+=\mathbb{A}_C^+$. Table \ref{table4}, contains all the obtained facet inequalities characterizing the classical set $\mathbb{A}_C^+$. Quantum violation is obtained only for three inequalities, $\mathcal{I}_5,\mathcal{I}_6$ and $\mathcal{I}_7$, in Table \ref{table4}, for which an extensive study is shown in Fig. \ref{fig224}.

\begin{table}[H]
    \centering
    \scriptsize
    \renewcommand{\arraystretch}{1.9} 
    \begin{tabular}{|c|c|c|}
    \hline
       Orbit & Facet- Inequalities & Maximum obtained quantum advantage  \\
        size  &  & in terms of $\frac{A_C^{\text{total}}}{A_{Q_4}^{\text{total}}}$ for fixed $\mathcal{S}$ \\
       \hline
       32  & $-p(1|1,1)+p(1|1,2)-p(2|1,1)+p(2|1,2)-p(3|1,1)+p(3|1,2) \leqslant 2A_2-1 $ & No advantage\\
       \hline
       96& $-2p(1|1,1)+p(1|2,1)+2p(1|1,2)-p(1|2,2)-p(2|2,1)+2p(2|1,2)-p(2|2,2) \leqslant 4A_1+2A_2-3$ & No advantage \\
       \hline


       96 & $\mathcal{I}_5= -p(1|1,1)+p(1|1,2)+p(2|1,1)-2p(2|2,1)+p(2|1,2)+2p(3|1,2)-2p(3|2,2) \leqslant 4A_1+2A_2-3$ & $\frac{A_C^{\text{total}}}{A_{Q_4}^{\text{total}}}=1.0066$ at $\mathcal{S}=1.0655$\\
       \hline
       
       192 &  $-p(1|2,1)-p(1|1,2)+2p(1|2,2)-2p(2|1,1)+p(2|2,1)+p(2|2,2)-2p(3|1,2)+2p(3|2,2) \leqslant 4A_1+2A_2-3$ & No advantage \\
       \hline


       96 & $-2p(1|1,2)+2p(1|2,2)-2p(2|1,2)+2p(2|2,2)-2p(3|1,1)+p(3|2,1)+p(3|2,2) \leqslant 4A_1+2A_2-3$ & No advantage\\
       \hline

       
       96 & $p(1|1,1)-2p(1|2,1)+p(1|1,2)+2p(3|1,2)-2p(3|2,2) \leqslant 4A_1+2A_2-3$ & No advantage\\
       \hline

       
       192 & $-2p(1|1,2)+2p(1|2,2)+p(2|1,1)-p(2|2,1)-p(2|1,2)+p(2|2,2)+p(3|1,1)-2p(3|2,1)-p(3|1,2)+$ & No advantage\\
        & $2p(3|2,2) \leqslant 4A_1+2A_2-3$ & \\
       \hline


       48& $-p(1|1,1)+p(1|2,2)+p(2|1,1)-p(2|2,1)-p(2|1,2)+p(2|2,2) \leqslant2A_1+2A_2-2$ & No advantage\\
       \hline


       24 & $p(1|1,1)-3p(1|2,1)-p(1|1,2)+3p(1|2,2)-2p(2|1,1)+2p(2|2,2)+p(3|1,1)-p(3|2,1)-3p(3|1,2)+$ & No advantage \\
       & $3p(3|2,2) \leqslant 4A_1+4A_2-4$ & \\
       \hline

       
       192 & $p(1|2,1)-2p(1|1,2)+p(1|2,2)+p(2|2,1)-p(2|2,2)+p(3|1,1)-p(3|2,1)-p(3|1,2)+p(3|2,2) \leqslant 4A_1+2A_2-3$ & No advantage \\
       \hline

       96 & $p(1|1,1)-3p(1|2,1)-p(1|1,2)+3p(1|2,2)-2p(2|2,1)+2p(2|2,2)-2p(3|1,1)+2p(3|2,2) \leqslant 4A_1+4A_2-4$ & No advantage\\
       \hline

       48 & $-p(1|1,1)+p(1|2,1)+p(1|1,2)-p(1|2,2)+p(2|2,1)-p(2|1,2)+p(3|2,1)-p(3|1,2) \leqslant2A_1+2A_2-2$ & No advantage \\
       \hline



       48 & $3p(1|1,1)-3p(1|2,1)-p(1|1,2)+p(1|2,2)+2p(2|1,1)-2p(2|2,2)+3p(3|1,1)-3p(3|2,1)-p(3|1,2)+$ & No advantage \\
       & $p(3|2,2) \leqslant 4A_1+4A_2-4$ & \\
       \hline

       
       96& $ \mathcal{I}_6= p(1|1,2)-p(1|2,2)-p(2|1,1)+p(2|2,1)+p(2|1,2)-p(2|2,2)+p(3|2,1)-p(3|2,2) \leqslant 2A_1+2A_2-2$ &  $\frac{A_C^{\text{total}}}{A_{Q_4}^{\text{total}}}=1.117$ at $\mathcal{S}=1.02$  \\
       \hline




       192 & $ -2p(1|1,1)+2p(1|2,1)+p(1|1,2)-p(1|2,2)+2p(2|2,1)-p(2|1,2)-p(2|2,2)+p(3|2,1)-p(3|1,2) \leqslant 2A_1+4A_2-3$ &No advantage\\
       \hline

       


       192 & $p(1|1,1)-3p(1|2,1)-p(1|1,2)+3p(1|2,2)-2p(2|2,1)+2p(2|2,2)+p(3|1,1)-p(3|2,1)-3p(3|1,2)+$ & No advantage \\
       & $3p(3|2,2) \leqslant 4A_1+4A_2-4$ & \\
       \hline

       192 & $\mathcal{I}_7= p(1|1,2)-p(1|2,2)-p(2|1,1)+2p(2|1,2)-p(2|2,2)-p(3|2,1)+p(3|1,2) \leqslant 2A_1+2A_2-2$ & $\frac{A_C^{\text{total}}}{A_{Q_4}^{\text{total}}}=1.0225$ at $\mathcal{S}=1.2166$ \\
       \hline

 
       96 & $p(1|1,1)-p(1|2,1)-3p(1|1,2)+3p(1|2,2)+p(2|1,1)-3p(2|2,1)-p(2|1,2)+3p(2|2,2)+p(3|1,1)-p(3|2,1)$ & No advantage\\
       & $-3p(3|1,2)+3p(3|2,2) \leqslant 4A_1+4A_2-4$ &  \\
       \hline

       96 &  $p(1|1,1)-p(1|2,1)-3p(1|1,2)+3p(1|2,2)-2p(2|1,1)+2p(2|2,2)-2p(3|2,1)+2p(3|2,2) \leqslant 4A_1+4A_2-4$ & No advantage\\
       \hline

       48& $p(2|1,1)-2p(2|1,2)+p(2|2,2)+p(3|1,1)-p(3|2,1)-p(3|1,2)+p(3|2,2) \leqslant 2A_1+2A_2-2$ & No advantage \\
       \hline

    \end{tabular}
\caption{We obtained $2210$ facet inequalities out of which $19$ are trivial and the rest can be grouped into $20$ different equivalence classes as listed above. Applied symmetry conditions are : $ 1\xlongleftrightarrow{x_1} 2$, $1\xlongleftrightarrow{x_2}2$ and $1\xlongleftrightarrow{z}2 $, $2\xlongleftrightarrow{z}3$, $1\xlongleftrightarrow{z}3$, $3\xlongleftrightarrow{z}4$, $ 1\xlongleftrightarrow{z}4$, $2\xlongleftrightarrow{z}4$, $(x_1\xlongleftrightarrow{}x_2) \wedge (A_1\xlongleftrightarrow{}A_2)$. }
\label{table4}
    
\end{table}

\begin{figure}[htbp!]
    \centering
    \begin{subfigure}[b]{0.5\textwidth}
        \centering
        \includegraphics[scale=0.6]{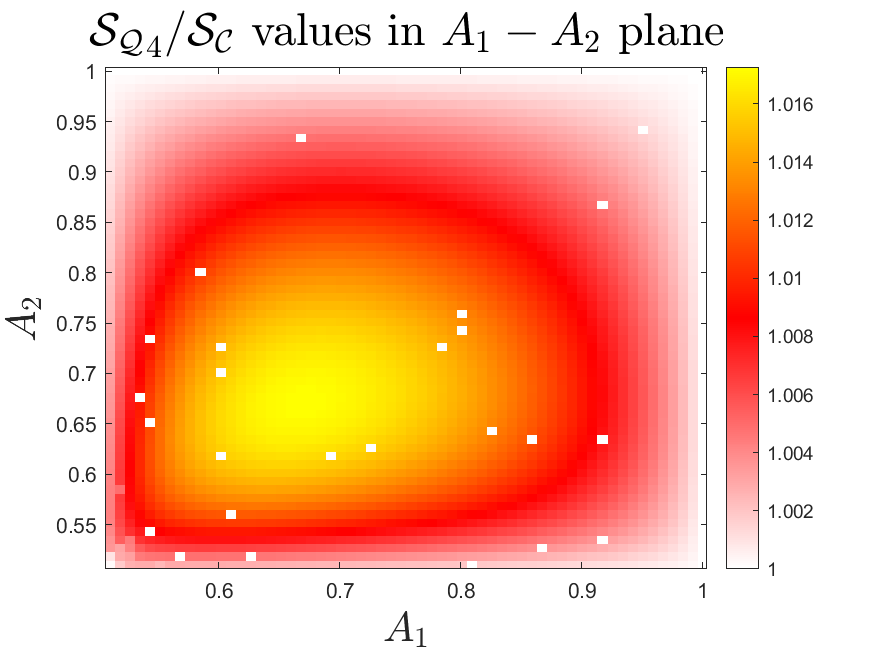}
        \caption{}
        \label{224subfigure1}
    \end{subfigure} 
    \hfill
    \begin{subfigure}[b]{0.49\textwidth}
        \centering
        \includegraphics[scale=0.6]{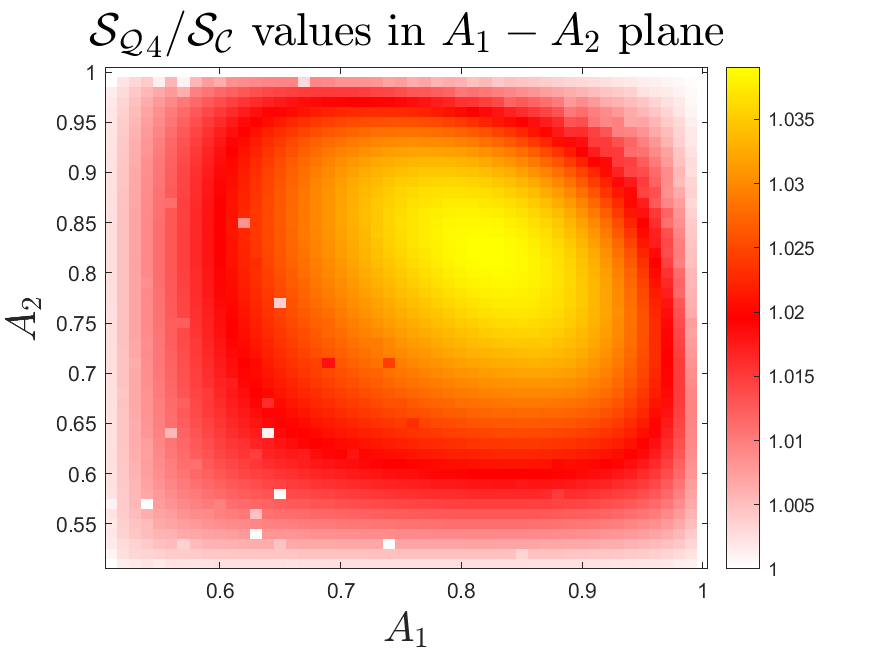}
        \caption{}
         \label{224subfigure 4}
    \end{subfigure}
    \hfill
    \begin{subfigure}[b]{0.5\textwidth}
        \centering
        \includegraphics[scale=0.6]{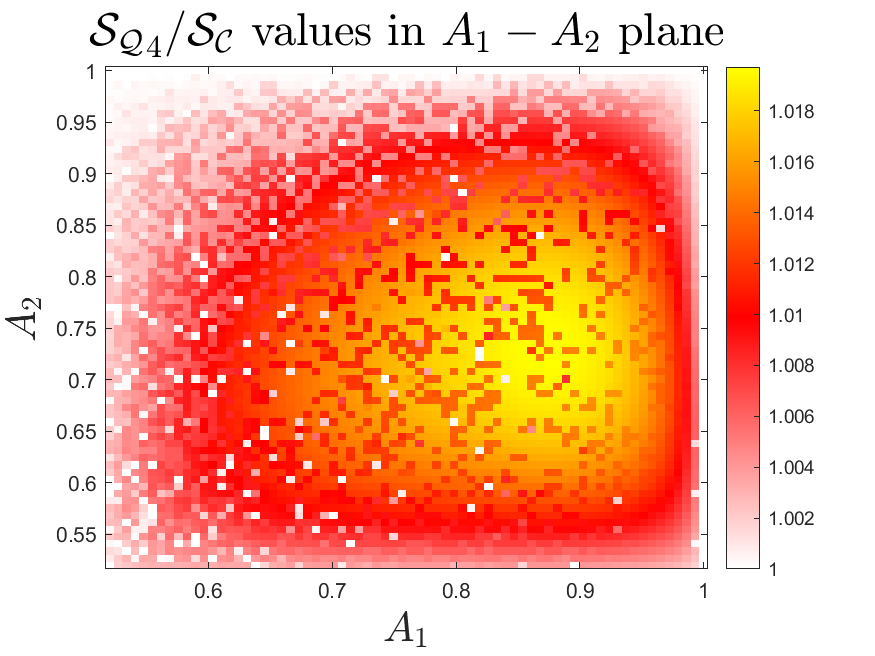}
        \caption{}
         \label{224subfigure 3}
    \end{subfigure}
    \caption{Quantum advantage in terms of $\mathcal{S}_{\mathcal{Q}_4}/\mathcal{S_C}$ for different values of $A_1$ and $A_2$ in $(2,2,4)$ scenario is presented. The Figures \eqref{224subfigure1},\eqref{224subfigure 4},\eqref{224subfigure 3}, respectively corresponds to the inequalities $\mathcal{I}_5,\mathcal{I}_6$ and $\mathcal{I}_7$ of Table \ref{table4}. 
    Here again, by considering values of $A_1$ and $A_2$ from the range [0.51,1] in small intervals, the optimal value of left-hand side of the corresponding inequality is computed using the Seesaw optimization technique. The classical bound is determined by substituting the same values of $A_1$ and $A_2$ into the classical bound of the inequality. Subsequently, the ratio of the quantum value to the classical bound is represented using the color-coded grid, where each grid point corresponds to a specific pair of $A_1$ and $A_2$ values. 
    }
    \label{fig224}
\end{figure}


\section{An explicit example of obtaining one vertex of $\mathbb{D}_C^+$ for $(3,2,2)$ distinguishability scenario}\label{app_3}
The extensive list of all the extremal classical strategies for the first sender, second sender and the receiver are available in Ref. \cite{ankush238_preparation_measurement_vertices_2025}. All of the enlisted strategies are obtained using `\texttt{lcon2vert}' package, readily available on `\texttt{Matlab}'. The first sender has $128$ extremal strategies, described in terms of the variables $\{p_e(m_1|x_1)\}, D_1$, where $x_1\in \{1,2,3\}$, can take three values and $m_1\in \{1,2,3,4\}$, can take four distinct values. Each of these strategies respects the constraints of (\ref{op4}),(\ref{en1}),(\ref{en2}) and (\ref{dis}). Each strategy is described by assigning certain values to the variables $\{p_e(m_1|x_1)\}, D_1$. Table \ref{322sender1table} describes one such extremal strategy for the first sender.

\begin{table}[H]

    \centering
    \begin{tabular}{|c|c|c|c|c|c|c|c|c|c|c|c|c|}
    \hline
       $p_e(1|1)$   & $p_e(2|1)$ & $p_e(3|1)$ & $p_e(4|1)$ & $p_e(1|2)$ & $p_e(2|2)$  & $p_e(3|2)$   &$p_e(4|2)$ & $p_e(4|3)$ & $p_e(4|3)$ & $p_e(4|3)$ & $p_e(4|3)$ & $D_1$    \\ \hline
       1 &  0 &  0 &   0 &  0 &  0&   0 &   1  &  0  &  0  &  0  &  1  &  2/3 \\ \hline
       \end{tabular}
       \caption{Extremal point of the encoding polytope associated with variables $\{p_e(m_1|x_1)\},D_1$.}
    \label{322sender1table} 
\end{table}

The second sender has six extremal  strategies, i.e., $p_e(m_2|x_2)\in\{0,1\}$, for $x_2 \in \{1,2 \}$ and $m_2\in \{1,2\}$. Note, only two choices for $m_2$ is sufficient here. Each of these strategies satisfies the constraints (\ref{op4}),(\ref{en1}),(\ref{en2}) and (\ref{dis}). Table \ref{322sender2table}, describes one such strategy for the second sender.

\begin{table}[htpb]
    \centering
    \begin{tabular}{|c|c|c|c|c|}
    \hline
     $p_e(1|1)$  & $p_e(2|1)$ & $p_e(1|2)$
     & $p_e(2|2)$ & $D_2$ \\
     \hline
    0 &1& 0& 1& 1/2 \\
    \hline
    \end{tabular}
    \caption{An extremal point of the encoding polytope associated with variables $\{p_e(m_2|x_2)\},D_2$ .}
\label{322sender2table}    
\end{table}

The extremal points of decoding polytope $\mathbb{M}$, i.e., the extremal strategies of receiver can be obtained by considering all possible deterministic output strategies, that is, $p_d(z|m_1,m_2) \in \{0,1\}$ such that $\sum_zp_d(z|m_1,m_2)=1$. There are a total $n_z^{n_{m_1}n_{m_2}}=2^{8}=256$ deterministic strategies. 
Table \ref{322measurementtable} presents one such extremal strategy for the receiver.

\begin{table}[htpb]
    \centering
    \scriptsize
    \renewcommand{\arraystretch}{1.9} 
    \begin{tabular}{|c|c|c|c|c|c|c|c|c|c|c|c|c|c|c|c|}
    \hline
      $p_d(1|1,1)$ & $p_d(1|1,2)$ & $p_d(1|2,1)$ & $p_d(1|2,2)$ & $p_d(1|3,1)$ & $p_d(1|3,2)$ & $p_d(1|4,1)$ & $p_d(1|4,2)$ &  $p_d(2|1,1)$ & $p_d(2|1,2)$ & $p_d(2|2,1)$ & $p_d(2|2,2)$ & $p_d(2|3,1)$ & $p_d(2|3,2)$ & $p_d(2|4,1)$ & $p_d(2|4,2)$ \\ \hline 

      1& 1& 1& 1& 1& 1& 1& 1& 0 & 0 & 0 & 0 & 0 & 0 & 0 & 0 \\ \hline
    \end{tabular}
    \caption{Extremal point of the decoding polytope associated with the variables $\{p_d(z|m_1,m_2)\}$. }

\label{322measurementtable}    
\end{table}
Now we are in a position to calculate an extremal point of $\mathbb{D}^+_C$. 
Obtain the operational probabilities $\{p(z|x_1,x_2)\}$ for all possible values of $z,x_1$ and $x_2$ via,
%
\bea
p(z|x_1,x_2) = \sum_{\substack{m_1 \in \{1,2,3,4\} \\ m_2 \in \{1,2\}}} p_d(z|m_1,m_2) p_e(m_1|x_1) p_e(m_2|x_2),  \label{eqc1}
\eea using the values of $p_d(z|m_1,m_2), p_e(m_1|x_1)$ and $p_e(m_2|x_2)$ from Table \ref{322sender1table},\ref{322sender2table} and \ref{322measurementtable}, respectively. The resulting extremal point of $\mathbb{D}_C^+$ is represented in Table \ref{tab:placeholder}.

\begin{table}[h]
    \centering
    \begin{tabular}{|c|c|c|c|c|c|c|c|c|c|c|c|c|c|}
    \hline
      $p(1|1,1)$ & $p(1|2,1)$ & $p(1|3,1)$ & $p(1|1,2)$ & $p(1|2,2)$ & $p(1|3,2)$ & $p(2|1,1)$ & $p(2|2,1)$ &  $p(2|3,1)$ & $p(2|1,2)$ & $p(2|2,1)$ & $p(2|3,2)$ & $D_1$ &$ D_2$ \\ \hline

      1  &  1  &  1 &   1  &  1  &  1       &  0    &     0  &       0   &      0  &       0    &     0   & 2/3  & 1/2 \\ \hline
    \end{tabular}
    \caption{Extremal point of $\mathbb{D}_C^+$, obtained from Table \ref{322sender1table},\ref{322sender2table} and \ref{322measurementtable} following \eqref{eqc1}.}
    \label{tab:placeholder}
\end{table}
Similarly, consider all other possible choices of combinations from the lists containing the extremal points for first sender, second sender and the receiver respectively, and generate the remaining vertices. Once all the vertices are obtained, one can generate the facet inequalities using `\texttt{polymake}' that are already listed in Table \ref{table3}. 

\end{document}